\documentclass{SciPost}

\hypersetup{
    colorlinks,
    linkcolor={red!50!black},
    citecolor={blue!50!black},
    urlcolor={blue!80!black}
}

\usepackage[bitstream-charter]{mathdesign}
\usepackage{upgreek}
\usepackage{graphicx}
\usepackage{dcolumn}
\usepackage{bm}
\usepackage{xcolor}
\usepackage{amsthm}
\usepackage{easyReview}
\usepackage{physics} 
\newtheorem{theorem}{Theorem}

\usepackage{diagbox}
\usepackage{array}

\def\ii{\mathrm{i}}
\def\ee{\mathrm{e}}
\def\no{\nonumber}
\global\long\def\ell#1{\theta_{#1}}
\global\long\def\bell#1{\tilde\theta_{#1}}

\DeclareSymbolFont{usualmathcal}{OMS}{cmsy}{m}{n}
\DeclareSymbolFontAlphabet{\mathcal}{usualmathcal}

\fancypagestyle{SPstyle}{
\fancyhf{}
\lhead{\colorbox{scipostblue}{\bf \color{white} ~SciPost Physics }}
\rhead{{\bf \color{scipostdeepblue} ~Submission }}

\fancyfoot[C]{\textbf{\thepage}}
}

\begin{document}

\pagestyle{SPstyle}

\begin{center}{\Large \textbf{\color{scipostdeepblue}{
Generalized Integrable Boundary States in XXZ and XYZ Spin Chains\\
}}}\end{center}

\begin{center}\textbf{
Xin Qian\textsuperscript{1,2} and
Xin Zhang\textsuperscript{3$\dagger$}
}\end{center}

\begin{center}
{\bf 1} School of Physics and Shing-Tung Yau Center, Southeast University, Nanjing 210096, China
\\
{\bf 2} Niels Bohr International Academy, Niels Bohr Institute, Blegdamsvej 17, 2100 Copenhagen, Denmark
\\
{\bf 3} Beijing National Laboratory for Condensed Matter Physics,
Institute of Physics, Chinese Academy of Sciences, Beijing 100190, China
\\[\baselineskip]
$\dagger$ \href{mailto:email2}{\small xinzhang@iphy.ac.cn}
\end{center}

\section*{\color{scipostdeepblue}{Abstract}}
\textbf{\boldmath{%
We investigate integrable boundary states in the anisotropic Heisenberg chain under periodic or twisted boundary conditions, for both even and odd system lengths. Our work demonstrates that the concept of integrable boundary states can be readily generalized. For the XXZ spin chain, we present a set of factorized integrable boundary states using the $\mathrm{KT}$-relation, and these states are also applicable to the XYZ chain. It is shown that a specific set of eigenstates of the transfer matrix can be selected by each boundary state, resulting in an explicit selection rule for the Bethe roots.
}}

\vspace{\baselineskip}

\noindent\textcolor{white!90!black}{%
\fbox{\parbox{0.975\linewidth}{%
\textcolor{white!40!black}{\begin{tabular}{lr}%
  \begin{minipage}{0.6\textwidth}%
    {\small Copyright attribution to authors. \newline
    This work is a submission to SciPost Physics. \newline
    License information to appear upon publication. \newline
    Publication information to appear upon publication.}
  \end{minipage} & \begin{minipage}{0.4\textwidth}
    {\small Received Date \newline Accepted Date \newline Published Date}%
  \end{minipage}
\end{tabular}}
}}
}


\vspace{10pt}
\noindent\rule{\textwidth}{1pt}
\tableofcontents
\noindent\rule{\textwidth}{1pt}
\vspace{10pt}


\section{Introduction}
\label{sec:intro}
Integrability is a powerful tool in the study of one-dimensional models. Integrable models possess enough conserved quantities, which allows us to derive its exact solutions. Since the pioneering work of Yang and Baxter \cite{PhysRevLett.19.1312,PhysRevLett.26.832}, many approaches have been developed to obtain exact solutions of integrable models, including the $T$–$Q$ relation \cite{baxter1982}, the algebraic Bethe ansatz \cite{Korepin_Bogoliubov_Izergin_1993,faddeev1996algebraicbetheansatzworks}, Sklyanin's separation of variables \cite{sklyanin1989}, the thermodynamic Bethe ansatz \cite{takahashi1999}, and the off‑diagonal Bethe ansatz \cite{Wang2015}. Integrable systems provide exact solvable laboratories for testing theoretical ideas, understanding non-perturbative effects, exploring the physics of many-body interactions, and developing powerful analytical techniques. This point of view is backed up by the revolutionary progress in cold atomic physics \cite{Guan_2013}.


A fundamental question in quantum many-body physics is whether and how a closed quantum system reaches thermal equilibrium. Quantum integrable systems are notable examples of systems that do not thermalize, where the presence of higher conserved charges is linked to the absence of full relaxation to equilibrium, and the correct statistical description of these systems is called the Generalized Gibbs Ensemble \cite{Rigol_2007, PhysRevLett.97.156403, Cramer_2008, Calabrese_2011, Caux_2012, Pozsgay_2013, Wouters_2014, Ilievski_2015}. The non-equilibrium dynamics of integrable model has been observed experimentally \cite{David2006quantum, Bloch_2008, Polkovnikov_2011}, and studied theoretically \cite{Calabrese_2016, Bastianello_2022}. A probably simplest protocol inducing out of equilibrium dynamics is the so called quantum quench, where a system is prepared in a well defined initial state and let evolve unitarily according to a known Hamiltonian \cite{Calabrese_2006,Calabrese_2007}. Among many analytical methods, quench action approach has been proven to be a powerful tool in the study of quantum quench in integrable models \cite{Caux_2013}. A key ingredient of this approach is the calculation of overlaps between the initial state and the eigenstates of the post-quench Hamiltonian. A subclass of integrable initial states, called integrable boundary states, was discussed in Ref. \cite{Piroli_2017}, and their definition is inspired by the classical work \cite{GHOSHAL_1994}. Among them, initial states that can be constructed from the solution of the boundary Yang–Baxter equation \cite{E_K_Sklyanin_1988} have received considerable attention due to their neat overlaps with Bethe states. For the isotropic XXX chain, the overlaps between these states and Bethe states are known, and the corresponding results can be generalized to higher rank cases \cite{gombor2024exactoverlapsallintegrable, gombor2025derivationsmpsoverlapformulas}. The key feature that shows up in the overlap is a universal term of the ratio of Gaudin determinant \cite{Pozsgay_2014, Brockmann_2014}, while the pre-factor is boundary dependent. 

Integrable boundary states are also important in the weak coupling limit of integrability within the $\mathrm{AdS/CFT}$ \cite{de_Leeuw_2015} correspondence. In this context, a key class of states known as Matrix Product States (MPS) arises \cite{Buhl_Mortensen_2016}. These MPS are integrable and can be interpreted as boundary states corresponding to a domain wall version of the $\mathcal{N}=4$ supersymmetric Yang-Mills ($\mathrm{SYM}$) theory. In many cases, the one-point functions of conformal operators—given by the overlap between these MPS and the Bethe eigenstates of the underlying integrable spin chain \cite{Minahan_2003, Beisert_2011} —can be computed in a closed form \cite{de_Leeuw_2016, de_Leeuw_2018}. One can refer to \cite{deleeuw2017introductionintegrabilityonepointfunctions} for a review. Alternatively, one can calculate such overlap using the representation theory of twisted Yangian \cite{de_Leeuw_2020, gombor2025exactoverlapsintegrablematrix}. Furthermore, three-point correlation functions involving two giant gravitons and a single-trace operator in $\mathcal{N}=4$ $\mathrm{SYM}$ can be evaluated using MPS of bond dimension two \cite{Jiang_2019, Jiang_2020}. The resulting expression generalizes the notion of $g$-functions in two-dimensional quantum field theory. A similar structure appears in the lower-dimensional $\mathrm{ABJM}$ theory. Here, the spectral problem is encoded in an alternating integrable spin chain with distinct representations on even and odd-site \cite{Minahan_2008, Bak_2008}. The three-point function of two giant gravitons and a single-trace operator has likewise been computed using MPS \cite{Yang_2022, yang2024integrableboundarystatesmaximal}. Moreover, one can formulate a $1/2$-$\mathrm{BPS}$ domain wall version of $\mathrm{ABJM}$ theory, whose corresponding one-point functions admit a closed-form expression via $\mathrm{MPS}$ techniques applied to the alternating spin chain \cite{Kristjansen_2022}. 

\begin{table}[htbp]
\centering
\begin{tabular}{|c|c|c|}
\hline
\diagbox{Boundary~~~\\state} {Twist $G$~~}& $\sigma^\alpha$ &  $\mathbb{I}$ \\[4pt]
\hline
 & & \\[-2pt]
$|\Psi_{+,e}\rangle$ &  $\left(\sum\limits_{\sigma=\pm}a_\sigma\,\lvert\varphi_{\alpha,\sigma}\rangle\otimes\lvert \varphi_{\alpha,-\sigma}\rangle\right)^{\otimes L/2}$   &  $~~{\begin{pmatrix}
    a\\
    b \\
    c \\
    d\\
\end{pmatrix}}^{\otimes L/2}$\\ 
&& \\
\hline 
 && \\
$|\Psi_{-,e}\rangle$ &   $\left(\sum\limits_{\sigma=\pm}a_\sigma\,\lvert\varphi_{\alpha,\sigma}\rangle\otimes \lvert\varphi_{\alpha,\sigma}\rangle\right)^{\otimes L/2}$  & ---\\ 
&& \\
\hline 
\end{tabular}
\caption{Two-site integrable boundary states $\lvert\Psi_{\pm,e}\rangle$ in even sites spin chain. Notation: $t(u)$ is corresponding transfer matrix, $\sigma^{\alpha}\lvert\varphi_{\alpha,\pm}\rangle=\pm\lvert\varphi_{\alpha,\pm}\rangle$, $\sigma^{\alpha}\ \mathrm{with}\ \alpha=1,2,3$ are Pauli matrices, $a_{\pm}\in\mathbb{C}$ and $L\in 2k,\ k\in\mathbb{N}$ is the number of sites. Here $G$ is the twist matrix, with $G=\mathbb{I}$ corresponding to the periodic boundary.}
\label{tab1}
\end{table}

A state $|\Psi\rangle$ satisfying the condition $t(u) |\Psi\rangle= t(-u) |\Psi\rangle$ is referred to as an integrable boundary state, where $t(u)$ denotes the transfer matrix. So far, most studies of integrable boundary states have focused on spin chains with an even number of sites. In this paper, we propose the boundary states $|\Psi_{\pm, s}\rangle$ satisfying
\begin{align}
t(u) |\Psi_{\pm, s}\rangle=\pm t(-u)|\Psi_{\pm, s}\rangle,\quad s=e,o, \label{def:BS}
\end{align}
where the index $s=e,o$ corresponds to even and odd site numbers, respectively. We obtain explicit expressions for several factorized boundary states in anisotropic spin chains with periodic and twisted boundaries. Our generalized integrable boundary states are summarized in Tables \ref{tab1} and \ref{tab2},  and they are valid for both the XXZ and XYZ chains. The novelty of this work lies in the following aspects. First, we establish the existence of integrable boundary states in systems with an odd site number. Second, in contrast to existing work focused only on the ``$+$'' branch in Eq. (\ref{def:BS}), we investigate both the ``$+$'' and ``$-$'' branches comprehensively. Furthermore, we extend the framework of integrable boundary states from the XXX/XXZ chain to the XYZ chain. In addition, we study XXZ/XYZ spin chains with both periodic and twisted boundary conditions.
All integrable boundary states presented in this work are factorized and derived from the $\mathrm{KT}$-relation, no matter what integrable boundary conditions are imposed or whether the number of sites is even or odd. We also discuss the non‑existence of certain boundary states in particular integrable models (see the blank entries in Tables \ref{tab1} and \ref{tab2}).
\begin{table}[htbp]
\centering
\begin{tabular}{|c|c|c|}
\hline
\diagbox{Boundary~~~ \\ state} {Twist $G$~~}& $\sigma^\alpha$ &  $\mathbb{I}$ \\
\hline
&& \\ $|\Psi_{+,o}\rangle$ &  $\prod\limits_{{\rm even}\, n}\sigma_n^\alpha\,{\begin{pmatrix}
    1\\
    a
\end{pmatrix}}^{\otimes L}$   &  ---\\ 
&&\\
\hline 
&&\\ $|\Psi_{-,o}\rangle$ &   ---  & ${\begin{pmatrix}
    1\\
    a
\end{pmatrix}}^{\otimes L}$\\ 
&& \\
\hline 
\end{tabular}
\caption{Integrable boundary states $\lvert\Psi_{\pm,o}\rangle$ in odd-site spin chain. Notation: $t(u)$ is corresponding transfer matrix, $\sigma^{\alpha}\ \mathrm{with}\ \alpha=1,2,3$ are Pauli matrices, $a\in\mathbb{C}$ and $L\in 2k+1,\ k\in\mathbb{N}$ is the number of sites.}
\label{tab2}
\end{table}

A key feature of integrable boundary states is that they select particular eigenstates (Bethe states) of the transfer matrix and Hamiltonian. The overlap between a Bethe state and an integrable boundary state may be nonzero only when the eigenvalue $\Lambda(u)$ of $t(u)$ satisfy $\Lambda(u)=\pm\Lambda(-u)$. For the XXZ and XYZ chains with integrable boundaries, the function $\Lambda(u)$ can be parameterized by the $T$-$Q$ relation \cite{baxter1982,faddeev1996algebraicbetheansatzworks,Cao_2013,Cao_2014,Wang2015,TakFad79,Slavnov_2020,Cao_2015}, where a set of Bethe roots is introduced. As a consequence, $\Lambda(u)=\pm\Lambda(-u)$ will leads to some constraints for the Bethe roots. In some cases, these constraints force the roots to form specific pairs. In others, no simple pairing pattern exists; instead, the roots should satisfy a set of non-trivial equations. In each case, we present the corresponding $T\mbox{-}Q$ relation, and then discuss the selection rule.

This paper is organized as follows. Section \ref{sec2} serves as a review of previous results in the periodic $\mathrm{XXZ}$ spin chain and sets up the notation. We establish a generalization of integrable boundary states and show that they exist in odd‑site periodic XXZ chains. In section \ref{sec3}, we lay out the explicit expression for generalized integrable boundary states and the corresponding selection rules in $\mathrm{XXZ}$ spin chain with twisted boundaries. Section \ref{sec4} focuses on the generalization of integrable boundary states to $\mathrm{XYZ}$ chain with periodic and twisted boundaries. Then, we give the corresponding selection rule of the integrable boundary states in Section \ref{sec5}. In section \ref{sec6}, we conclude our results and provide a outlook for future research. Appendix \ref{Appendix_A} is about the convention of elliptic functions. Appendix~\ref{Appendix_B} presents the construction of the dual states for the integrable boundary state. We prove the equivalence of $\mathrm{KYB}$ equation and $\mathrm{BYB}$ equation for both $\mathrm{XXZ}$ and $\mathrm{XYZ}$ chain in Appendix \ref{Appendix_C}. The component expressions of the $\mathrm{KT}$-relation are given in Appendix \ref{Appendix_D}.

\section{Integrable boundary states in the periodic XXZ chain}
\label{sec2}
So far, most studies of integrable boundary states have focused on periodic systems. In this section, using the periodic XXZ spin chain as an example, we recall some fundamental concepts of integrable models, review existing research results on integrable boundary states, and then generalize the definition of integrable boundary states.

\subsection{Integrability and exact solutions of the periodic XXZ chain}
The Hamiltonian of the periodic XXZ chain is
\begin{equation}
H=\sum_{i=1}^L\,\sum_{\alpha=x,y,z}J^{\alpha}\sigma_i^{\alpha}\sigma_{i+1}^{\alpha},\quad J^x=J^y=1,\quad J^z=\cosh\eta,\quad \sigma_{L+1}^{\alpha}\equiv\sigma_1^{\alpha}.
\label{ham}
\end{equation}
The $R$-matrix corresponding to the XXZ model reads 
\begin{equation}
    R_{0,j}(u)=\frac{1}{2}\left[\frac{\sinh(u+\eta)}{\sinh\eta}(1+\sigma_0^z\sigma_j^z)+\frac{\sinh u}{\sinh\eta}(1-\sigma_0^z\sigma_j^z)\right]+\frac{1}{2}(\sigma_0^x\sigma_j^x+\sigma_0^y\sigma_j^y),
    \label{Rmatxxz}
\end{equation}
where $0$ index denotes the auxiliary space and $j$ denotes the $j$-th site physical space. Lax operator is defined by a shift of spectral parameter of the $R$ matrix $\mathcal{L}_{0,j}(u)=R_{0,j}(u-\eta/2)$. The Lax operator satisfies the crossing relation
\begin{equation}
    \mathcal{L}_{0,1}(u)=-\sigma_1^y\mathcal{L}_{0,1}^{\mathrm{t}}(-u)\sigma_1^y,
\end{equation}
with $\mathrm{t}$ denoting transpose. One can build the monodromy matrix from the Lax operator
\begin{equation}
    T_0(u)=\mathcal{L}_{0,L}(u)\dots \mathcal{L}_{0,1}(u)=\begin{pmatrix}
        A(u)&B(u)\\C(u)&D(u)
    \end{pmatrix},\label{mon:matrix:xxz}
\end{equation}
and the $\mathrm{RTT}$-relation is guaranteed by the fundamental relation
\begin{equation}
    R_{0,\bar{0}}(u-v)\mathcal{L}_{0,j}(u)\mathcal{L}_{\bar{0},j}(v)=\mathcal{L}_{\bar{0},j}(v)\mathcal{L}_{0,j}(u)R_{0,\bar{0}}(u-v).\label{RLL}
\end{equation}
The trace of the monodromy matrix in the auxiliary space gives the transfer matrix
\begin{equation}
t(u)=\mathrm{Tr}_0(\mathrm{T}_0(u))=A(u)+D(u).\label{Transfer:PBC}
\end{equation} 
The logarithmic derivative of the transfer matrix yields the Hamiltonian in Eq. (\ref{ham})
 \begin{equation}
     H=2\sinh\eta\frac{\partial\ln t(u)}{\partial u}\Big\lvert_{u=\frac{\eta}{2}} - L\cosh\eta\times\mathbb{I}.
 \end{equation}
The eigenvalue of the transfer matrix, denoted as $\Lambda(u)$, can be given by the $T$-$Q$ relation \cite{baxter1982}
\begin{equation}
\Lambda(u)Q(u)=a(u)Q(u-\eta)+d(u)Q(u+\eta),\label{TQ:Periodic}
\end{equation}
where 
\begin{align}
d(u)=a(u-\eta)=\left[\frac{\sinh(u-\frac{\eta}{2})}{\sinh\eta}\right]^L,\label{def:ab}
\end{align}
and 
\begin{equation}
    Q(u)=\prod_{j=1}^M\frac{\sinh(u-u_j)}{\sinh\eta}.\label{Q:PBC}
\end{equation}
The integer $M$ in (\ref{Q:PBC}) takes values from $0$ to $L$, which can be seen as the magnon number. 
The Bethe roots $\{u_1,\dots,u_M \}$ must satisfy the following Bethe equations
\begin{equation}\label{BAE:PBC}
\left[\frac{\sinh(u_j+\frac{\eta}{2})}{\sinh(u_j-\frac{\eta}{2})}\right]^L=\prod_{k\neq j}^M\frac{\sinh(u_j-u_k+\eta)}{\sinh(u_j-u_k-\eta)},\quad l=1,\dots,M.
\end{equation}


One can also construct the eigenvectors of the transfer matrix by the algebraic Bethe ansatz approach \cite{Korepin_Bogoliubov_Izergin_1993}, specifically as follows
\begin{align}
\lvert \mathbf{u}\rangle=B(u_1)\dots B(u_M)\lvert0\rangle,\\
\langle \mathbf{u}|=\langle 0|C(u_1)\dots C(u_M),
\end{align}
where $\lvert0\rangle$ is the all spin up state and $\mathbf{u}=\{u_1,\ldots,u_M\}$ are given by Eq. (\ref{BAE:PBC}). The energy corresponding to Bethe state $|\mathbf{u}\rangle$ is given by
\begin{equation}
    E(\mathbf{u})=2\sum_{j=1}^{M}[g(u_j-\tfrac{\eta}{2})-g(u_j+\tfrac{\eta}{2})]+Lg(\eta),\quad g(u)=\sinh\eta\coth u.
\end{equation}
For the later calculation of the overlap between integrable boundary states and Bethe states, we need to introduce norm of Bethe states, denote 
\begin{equation}
    \mathcal{K}_s(u)=\coth(s\eta+u)-\coth(-s\eta+u),\quad  f(u,v)=\frac{\sinh(u-v-\eta)}{\sinh(u-v)},
\end{equation}
then the norm of Bethe states takes the form
\begin{equation}
    \langle\mathbf{u}\lvert\mathbf{u}\rangle=\big(\sinh(\eta)\big)^M\prod_{j=1}^{M}a(u_j)d(u_j)\prod_{k<j}^{M}f(u_j,u_k)\det G
\label{norm_peri}
\end{equation}
where the Gaudin matrix $G$ takes the following form 
\begin{equation}
    G_{ik}=-\delta_{ik}\Big(L\mathcal{K}_{1/2}(u)+\sum_{j=1}^M\mathcal{K}_{-1}(u_i-u_j)\Big)+\mathcal{K}_{-1}(v_i-v_k),\quad i,k=1,\dots M.    
\label{Bethe_norm_pxxz}
\end{equation}
\subsection{Integrable boundary states}\label{sec:xxz:bs}
Following Refs. \cite{Jiang_2020_overlap, Gombor_2021}, a state $|\Psi \rangle$ is called an \emph{integrable boundary state} if it satisfies the following constraint
\begin{equation}
    [t(u)-t(-u)]\lvert\Psi\rangle=0.
    \label{intconI}
\end{equation}

It turns out that a very large class of states satisfies Eq. \eqref{intconI} \cite{Caetano_2022, He_2023}. Among these, there exists a subclass that exhibits a factorized structure and can be constructed using the boundary $K$-matrix. In this work, we focus on this specific subclass.

The generic $K$-matrix of the XXZ model can be written in the following form \cite{Vega_1993}
\begin{align}
K(u)&=\frac{\alpha_1}{k_1(u)}\sigma^x+\frac{\alpha_2}{k_2(u)}\sigma^y+\frac{\alpha_3}{k_3(u)}\sigma^z+\frac{\alpha_4}{k_4(u)}\mathbb{I}\nonumber\\
&=\left(
\begin{array}{cc}
k_{11}(u) & k_{12}(u) \\
k_{21}(u) & k_{22}(u) \\
\end{array}
\right),\label{K:XXZ}
\end{align}
where $\alpha_1,\alpha_2,\alpha_3$ and $\alpha_4$ are free parameters and 
\begin{align}
k_1(u)=k_2(u)=1,\quad  k_3(u)=2\cosh u,\quad k_4(u)=2\sinh u.
\end{align}
The $K$-matrix Eq. (\ref{K:XXZ}) satisfies the boundary Yang-Baxter equation
\begin{equation}
    R_{1,2}(u-w)K_1(u)R_{1,2}(u+w)K_2(w)=K_2(w)R_{1,2}(u+w)K_1(u)R_{1,2}(u-w).
    \label{BYBeq}
\end{equation}

\begin{theorem}\label{thm0}
For an even number of sites system, we can construct the following integrable boundary state
\begin{equation}
    \lvert \Psi_0\rangle=\lvert\psi_0\rangle_{1,2}^{\otimes L/2},
    \label{intxxzp}
\end{equation}
where $\lvert\psi_0\rangle_{1,2}$ is a two-site state derived from the $K$-matrix
\begin{equation}
\lvert\psi_0\rangle_{1,2}=\sum_{i,j=1}^2[K(-\tfrac{\eta}{2})\sigma^y]_{ij}\lvert i\rangle_1\otimes\lvert j\rangle_2.\label{two:site:state}
\end{equation}
\end{theorem}
\begin{proof}
    Denote the vector version of the boundary $K(u)$ matrix as $\vec{K}'(u)$, which takes the form
\begin{equation}
    \vec{K}'(u)=\sum_{i,j}k'_{ij}(u)e_i\otimes e_j,\quad k'_{ij}=\left[K(-u)\sigma^y\right]_{ij},
\end{equation}
where the vector $e_1=(1,\,0)$ and $e_2=(0,\,1)$. We get a similar $\mathrm{KYB}$ equation (see Appendix \ref{Appendix_B} for more details)
\begin{equation}
    \mathcal{L}_{1,3}(u-v)\mathcal{L}_{1,4}(u+v)\vec{K}'_{1,2}(u)\vec{K}'_{3,4}(v+\tfrac{\eta}{2})=\mathcal{L}_{2,4}(u-v)\mathcal{L}_{2,3}(u+v)\vec{K}'_{1,2}(u)\vec{K}'_{3,4}(v+\tfrac{\eta}{2}).
    \label{KYBEQ2}
\end{equation}
From the above KYB equation Eq. (\ref{KYBEQ2}), one will get the $\mathrm{KT}$-relation
\begin{equation}
    \sigma_0^y K_0^{\mathrm{t}}(-u)\sigma_0^yT_0(u)\lvert\Psi_0\rangle= T_0(-u)\sigma_0^y K_0^{\mathrm{t}}(-u)\sigma_0^y\lvert\Psi_0\rangle.
    \label{KT3}
\end{equation}
Assuming that all eigenstates of the reflection matrix $K(u)$ are non-degenerate, taking the trace over the auxiliary space on both sides of \eqref{KT3} yields directly $[t(u)-t(-u)]\lvert\Psi_0\rangle=0$.
\end{proof}
 Obviously, the state $|\Psi_{0}\rangle$ introduced in Eq. \eqref{intxxzp} is a boundary state for a chain with an even $L$, and has a factorized structure. 
Because the $K$–matrix contains four independent parameters, the two‑site state
$|\psi_{0}\rangle_{1,2}$ in \eqref{two:site:state} can represent an arbitrary
vector of the tensor‑product space $V_{1}\!\otimes\! V_{2}$ by an appropriate choice of the boundary parameters $\{\alpha_1,\ldots,\alpha_4\}$.  One can also construct the bra vector $\bra{\Psi_0}$ with the same technique, see Appendix \ref{Appendix_B}.

It should be noted that the KT relation \eqref{KT3} yields a set of identities regarding the matrix elements of $T(-u)$ and $T(u)$. These identities are essential for proving the existence of boundary states in spin chains with an odd number of sites. More details are given in Appendix \ref{Appendix_D}.

\paragraph{Selection rule}

A characteristic feature of an integrable boundary state is that it selects only a subset of Bethe states. Consider the overlap between a Bethe state
$|\mathbf{u}\rangle$ and the boundary state $|\Psi\rangle$.  From Eq. \eqref{intconI} we obtain  
$$\langle \mathbf{u}|[t(u)-t(-u)]|\Psi\rangle
= [\Lambda(u)-\Lambda(-u)]\langle \mathbf{u}|\Psi\rangle=0. 
$$
This equation implies that the overlap $\langle\mathbf{u}|\Psi\rangle$
can be non‑zero only when  
$\Lambda(u)=\Lambda(-u)$.
Thus the boundary state imposes a selection rule on the
configurations of Bethe roots.


A sufficient and necessary condition for $\Lambda(u)=\Lambda(-u)$ is that the Bethe roots in Eqs. \eqref{TQ:Periodic} and \eqref{BAE:PBC} satisfy the following selection rules
\begin{equation}
\{u_1,\dots,u_M\}=\{-u_1+\ii\pi l_1,\dots,-u_M+\ii\pi l_M\},\quad l_i\in\mathbb{Z}, \label{Selection:PBC}
\end{equation}
with even sites $L$. We should emphasize that given the above selection rules, the overlap $\langle\mathbf{u}\lvert\Psi\rangle$ may still be zero in some specific cases, and this comment applies to the rest of paper.

The overlap $\langle\mathbf{u}|\Psi\rangle$ plays a crucial role in the study of non‑equilibrium dynamics of integrable systems. For the state $|\Psi_0\rangle$ introduced in Eq. \eqref{intxxzp}, it has been shown that, when the Bethe roots satisfy the selection rule \eqref{Selection:PBC}, the overlap $\langle\mathbf{u}|\Psi_0\rangle$ can be expressed as a Gaudin‑type determinant \cite{Pozsgay_2014, Brockmann_2014}.
\paragraph{Boundary state overlaps}
For the overlap between integrable boundary states and Bethe state, denote
\begin{equation}
\begin{split}
    G^+_{ik}&=-\delta_{ik}\Big(L\mathcal{K}_{1/2}(u_j)+\sum_{j=1}^{M/2}\big(\mathcal{K}_{-1}(u_i-u_j)+\mathcal{K}_{-1}(u_i+u_j)\big)\Big)+(\mathcal{K}_{-1}(u_i-u_k)+\mathcal{K}_{-1}(u_i+u_k)),\\
    G^{-}_{ik}&=-\delta_{ik}\Big(L\mathcal{K}_{1/2}(u_j)+\sum_{j=1}^{M/2}\big(\mathcal{K}_{-1}(u_i-u_j)+\mathcal{K}_{-1}(u_i+u_j)\big)\Big)+(\mathcal{K}_{-1}(u_i-u_k)-\mathcal{K}_{-1}(u_i+u_k)),
\label{gaulike_det}
\end{split}
\end{equation}
where $i,k=1,\dots,M/2$. Choosing the following parametrization of the boundary $K$-matrix
\begin{equation}
\begin{split}
    \alpha_1=\sinh 2u \cosh \theta,&\quad 
    \alpha_2=i\sinh 2u \sinh \theta\\
    \alpha_3=2\sinh 2u \cosh\alpha \sinh\beta,&\quad 
    \alpha_4=2\sinh 2u\sinh\alpha\cosh\beta,
\end{split}    
\end{equation}
then the overlap between integrable boundary states (\ref{intxxzp}) and Bethe states take the following form\cite{Pozsgay_2018}
\begin{equation}
    \frac{\langle\mathbf{u}\lvert \Psi_0\rangle}{\sqrt{\langle\mathbf{u}\lvert\mathbf{u}\rangle}}=\lvert \ee^{\theta (L-2M)}\lvert^{1/2}(\sinh\eta)^{L/2}\prod_{j=1}^{M/2}\big(y_l(\alpha,\beta,u_j)\big)^{1/2}\sqrt{\frac{\det G^{+}}{\det G^{-}}}
\end{equation}
where 
\begin{equation}
y_l(\alpha,\beta,u)=\frac{z_s(\alpha,u)z_s(\alpha^*,u)z_c(\beta,u)z_c(\beta^*,u)}{z_s(\eta/2,u)z_s(0,u)z_c(\eta/2,u)z_c(0,u)},
\label{yl_def}
\end{equation}
with
\begin{equation}
    z_s(\kappa,\lambda)=\sinh(\lambda+\kappa)\sinh(\lambda-\kappa),\quad z_c(\kappa,\lambda)=\cosh(\lambda+\kappa)\cosh(\lambda-\kappa). 
\end{equation}


\subsection{Generalized boundary states}

In this paper, we aim to generalize the concept of an integrable boundary state in Eq. (\ref{intconI}). The generalization manifests in the following aspects:
(1) the constraint for the boundary state can be extended from $t(u)|\Psi\rangle = t(-u)|\Psi\rangle$ to $t(u)|\Psi\rangle = \pm t(-u)|\Psi\rangle$;
(2) boundary states can also exist in systems with an odd number of sites. For convenience, we adopt the following definition for the extended boundary state:
\begin{align}
\label{intdef}
t(u) |\Psi_{\pm, s}\rangle = \pm t(-u) |\Psi_{\pm, s}\rangle, \quad s = e, o,
\end{align}
where $e$ and $o$ denote systems with an even and odd total number of sites, respectively.

In Section \ref{sec:xxz:bs}, we have demonstrated the existence of $|\Psi_{+,e}\rangle$ (see Eq. (\ref{intxxzp})) in periodic XXZ chain. We now turn to the investigation of other possible extended boundary states in the periodic XXZ chain.
\begin{theorem}\label{thm1}
For a periodic XXZ chain with an odd $L$, we can construct the following factorized integrable boundary state
\begin{equation}
\lvert\Psi_{-,o}\rangle=\begin{pmatrix}1\\a\end{pmatrix}^{\otimes L}
,\quad a\in \mathbb{C}.\label{Boundary:State:Twist}
\end{equation}
\end{theorem}
\begin{proof}
Let us define 
\begin{align}
\mathcal{L}_{0,L}(u)\cdots \mathcal{L}_{0,2}(u)=\begin{pmatrix}\mathfrak{A}(u)&\mathfrak{B}(u)\\\mathfrak{C}(u)&\mathfrak{D}(u)\end{pmatrix},\quad 
\mathcal{L}_{0,1}(u)=\begin{pmatrix}\mathfrak{a}(u)&\mathfrak{b}(u)\\\mathfrak{c}(u)&\mathfrak{d}(u)\end{pmatrix}.
\end{align}
The transfer matrix defined in Eq. \eqref{Transfer:PBC} can be expressed as
\begin{align}
t(u)=\mathfrak{A}(u)\mathfrak{a}(u)+\mathfrak{B}(u)\mathfrak{c}(u)+\mathfrak{C}(u)\mathfrak{b}(u)+\mathfrak{D}(u)\mathfrak{d}(u).
\end{align}
Acting the transfer matrix on the state $|\Psi_{-,o}\rangle$, one can derive 
\begin{align}
t(\pm u)|\Psi_{-,o}\rangle=\left(
\begin{array}{c}
\left[\frac{\sinh(\pm u+\frac{\eta}{2})}{\sinh\eta}\ \mathfrak{A}(\pm u)+a\mathfrak{B}(\pm u)+\frac{\sinh(\pm u-\frac{\eta}{2})}{\sinh\eta}\mathfrak{D}(\pm u)\right]\lvert\Phi_{0}\rangle\\ [6pt]
\left[(a\frac{\sinh(\pm u-\frac{\eta}{2})}{\sinh\eta}\mathfrak{A}(\pm u)+\mathfrak{C}(\pm u)+a\frac{\sinh(\pm u+\frac{\eta}{2})}{\sinh\eta}\mathfrak{D}(\pm u)\right]\lvert\Phi_{0}\rangle
\end{array}
\right),
\end{align}
where 
\begin{align}
\lvert\Phi_{0}\rangle=\begin{pmatrix}1\\a\end{pmatrix}_2\otimes\begin{pmatrix}1\\a\end{pmatrix}_3\otimes\dots\otimes\begin{pmatrix}1\\a\end{pmatrix}_{L-1}\otimes\begin{pmatrix}1\\a\end{pmatrix}_L.\label{Def:Psi0p}
\end{align}
We see that the state $|\Phi_0\rangle$ defined in Eq. (\ref{Def:Psi0p}) can be seen as a two-site boundary state, which can be obtained from Eq. (\ref{two:site:state}) by setting the $K$-matrix parameters to the following specific values
\begin{align}
\alpha_1=1,\quad \alpha_2=-\frac{\ii (a^2+1) k_2(\frac{\eta }{2})}{(a^2-1)k_1(\frac{\eta}{2})},\quad \alpha_3=\frac{2ak_3(\frac{\eta}{2})}{(a^2-1)k_1(\frac{\eta}{2})},\quad \alpha_4=0.
\label{Kmatsol_1}
\end{align}
The operators $\mathfrak{A}(u)$, $\mathfrak{B}(u)$, $\mathfrak{C}(u)$, $\mathfrak{D}(u)$ and the state $|\Phi_0\rangle$ now satisfy Eq. (\ref{TT:Relation}).
Using Eqs. (\ref{Kmatsol_1}) , (\ref{TT:Relation}) and the following identities 
\begin{align}
&\frac{\sinh(u\pm\tfrac{\eta}{2})}{\sinh\eta}=\frac{1}{2}\left[\frac{k_4(u)}{k_4(\frac{\eta }{2})}\pm \frac{k_3(u)}{k_3(\frac{\eta }{2})}\right],\\
k_1(u)=k_2&(u)=1,\quad k_3(-u)=k_3(u),\quad 
k_4(-u)=-k_4(u),
\end{align}
we can finally prove that 
\begin{align}
t(u) \lvert\Psi_{-,o}\rangle = -t(-u) \lvert\Psi_{-,o}\rangle.
\end{align}
\end{proof}
A nonzero overlap between $\lvert\Psi_{-,o}\rangle$ and Bethe state $|\mathbf{u}\rangle$ requires $\Lambda(u)=-\Lambda(-u)$,  implying the following selection rule for the Bethe roots ($L$ is odd)
\begin{equation}
\{u_1,\dots,u_M\}=\{-u_1+\ii\pi l_1,\dots,-u_M+\ii\pi l_M\},\quad l_i\in\mathbb{Z}. \label{Selection:PBCO}
\end{equation}
The overlap between $\lvert\Psi_{-,o}\rangle$ and Bethe state yields 
\begin{equation}
    \frac{\langle\mathbf{u}\lvert \Psi_{-,o}\rangle}{\sqrt{\langle\mathbf{u}\lvert\mathbf{u}\rangle}}=a^{2M}\Big(\prod_{j=1}^{M/2}\big(\tanh^2 u_j\tanh(u_j+\eta/2)\tanh(u_j-\eta/2)\big)\Big)\sqrt{\frac{\det G^{+}}{\det G^{-}}},
\end{equation}
where $G^{+}, G^-$ are Gaudin like determinant in Eq. (\ref{gaulike_det}).

We have shown that the boundary states $|\Psi_{+,e}\rangle$ and $|\Psi_{-,o}\rangle$ exist in the periodic XXZ chain. 
A natural question is whether one can also construct $|\Psi_{-,e}\rangle$ and $|\Psi_{+,o}\rangle$.
From the $T$-$Q$ relation Eq. (\ref{TQ:Periodic}),  we can get
\begin{align}
\begin{aligned}\label{Leading:Term:PBC}
\lim_{u\to +\infty}\frac{\Lambda(u)}{(2\sinh\eta)^{-L}\ee^{L u}}&=2\cosh\left((\tfrac{L}{2}-M\right)\eta),\\
\lim_{u\to +\infty}\frac{\Lambda(-u)}{(2\sinh\eta)^{-L}\ee^{L u}}&=2(-1)^L\cosh\left((\tfrac{L}{2}-M)\eta\right).
\end{aligned}
\end{align}
For an even $L$, it follows directly from Eq. (\ref{Leading:Term:PBC}) that the condition $\Lambda(u)+\Lambda(-u)=0$ can not be satisfied. 
If a state $|\Psi_{-,e}\rangle$ nevertheless existed, its overlap with every eigenstate of the transfer matrix $t(u)$ would have to vanish, forcing $|\Psi_{-,e}\rangle$ to be the trivial state.
Similarly, for a periodic system with odd $L$, no nontrivial $|\Psi_{+,o}\rangle$ can exist.

\section{Integrable boundary states in the twisted XXZ chain}
\label{sec3}
It is well known that the XXZ model with a twisted boundary remains integrable. Let $G$ denote the operator that implements the twisted boundary. When the matrix entries of $G$ satisfy
\begin{align}
G_{11}G_{12}=G_{11}G_{21}=G_{22}G_{12}=G_{22}G_{21}=0,
\label{twRGG}
\end{align}
the $R$-matrix in Eq. (\ref{Rmatxxz}) satisfies \begin{align}[R(u),\,G\otimes G]=0.\label{RGG}
\end{align}

Define the transfer matrix 
\begin{align}
t(u)={\rm Tr}_0(G_0T_0(u))\equiv {\rm Tr}_0(T_0(u)G_0),\label{twist:transfer}
\end{align}
where $T_0(u)$ is given by Eq. \eqref{mon:matrix:xxz}. Using the RLL relation \eqref{RLL} together with Eq. \eqref{RGG}, we can prove the commutation relation $[t(u),\,t(v)]=0$. This demonstrates that the twist preserves the integrability of the system.

In Ref. \cite{Gombor_2021}, the authors generalized the KT relation and presented an approach to construct a two-site integrable boundary state $\lvert\Psi_{+,e}\rangle$ in a spin chain with twisted boundary conditions (for even $L$). The key idea of the construction is that the following two-site state
\begin{align}
\lvert \Psi_0\rangle=\lvert\psi_0\rangle_{1,2}^{\otimes L/2},\quad 
\lvert\psi_0\rangle_{1,2}=\sum_{i,j=1}^2[K(\tfrac{\eta}{2})\sigma^y]_{ij}\lvert i\rangle_1\otimes\lvert j\rangle_2,\label{two:site:state:twist}
\end{align}
satisfies the integrable boundary condition
\begin{equation}
 [t(u)-t(-u)]\lvert\Psi_{0}\rangle=0, \label{intypei}
\end{equation}
provided that the twist matrix $G$ and the $K$-matrix obey the compatibility relation
\begin{equation}
\sigma^y K^{\mathrm{t}}(-u)\sigma^yG=G\sigma^y K^{\mathrm{t}}(-u)\sigma^y. \label{KTcompatI}
\end{equation}
We realize that this method remains applicable for constructing a two-site integrable boundary state $\lvert\Psi_{-,e}\rangle$ by replacing \eqref{KTcompatI} with 
\begin{equation}
\sigma^y K^{\mathrm{t}}(-u)\sigma^yG=-G\sigma^y K^{\mathrm{t}}(-u)\sigma^y. \label{KTcompatII}
\end{equation}

\subsection{Construction of \texorpdfstring{$|{\Psi_{+,e}}\rangle$}{Psi\_{+,e}}}\label{sec:+e}

The compatibility relation \eqref{KTcompatI} gives 
\begin{align}
&\left(
\begin{array}{cc}
G_{11} k_{22}(-u)-G_{21} k_{12}(-u) & G_{12} k_{22}(-u)-G_{22} k_{12}(-u) \\
 G_{21} k_{11}(-u)-G_{11} k_{21}(-u) & G_{22} k_{11}(-u)-G_{12} k_{21}(-u) \\
\end{array}
\right)\no\\
=&\left(
\begin{array}{cc}
 G_{11} k_{22}(-u)-G_{12} k_{21}(-u) & G_{12} k_{11}(-u)-G_{11} k_{12}(-u) \\
 G_{21} k_{22}(-u)-G_{22} k_{21}(-u) & G_{22} k_{11}(-u)-G_{21} k_{12}(-u) \\
\end{array}
\right).\label{Constraint:plus}
\end{align}
Since Eq. \eqref{Constraint:plus} must hold for arbitrary $u$, with the help of the expression of $k_{ij}(u)$ in Eq. \eqref{K:XXZ}, one can get the following equations
\begin{align}
\begin{aligned}\label{Eq:plus}
&G_{21}k_{12}(-u)=G_{12}k_{21}(-u),\quad G_{12}[k_{22}(-u)-k_{11}(-u)]=0, \\
&k_{12}(-u)(G_{11}-G_{22})=0,\quad 
G_{21}[k_{11}(-u)-k_{22}(-u)]=0,\\
&k_{21}(-u)(G_{11}-G_{22})=0.
\end{aligned}
\end{align}
In the case $G_{12}=G_{21}=0$ and $G_{11},G_{22}\neq 0$, we get 
\begin{align}
k_{12}(-u)=k_{21}(-u)=0.
\end{align}
Therefore, the local two-site state reads

\begin{align}
\lvert\psi_0\rangle_{1,2}
\propto\left(\begin{array}{cccc}
0&
-k_{11}(\frac{\eta}{2})&
k_{22}(\frac{\eta}{2})&
0
\end{array}\right)^{\mathrm{t}},
\end{align}
with $\mathrm{t}$ denoting a transpose. 
Now we consider the case $G_{11}=G_{22}=0$ and $G_{12},G_{21}\neq 0$, we get  
\begin{align}
k_{11}(-u)=k_{22}(-u),\quad k_{21}(-u)=\frac{G_{21}}{G_{12}}k_{12}(-u).
\end{align}
The local two-site state thus reads

\begin{align}
\lvert\psi_0\rangle_{1,2}\propto\left(\begin{array}{cccc}
k_{12}(\frac{\eta}{2})&
-k_{11}(\frac{\eta}{2})&
k_{11}(\frac{\eta}{2})&
-\frac{G_{21}}{G_{12}}k_{12}(\frac{\eta}{2})
\end{array}\right)^{\mathrm{t}}.
\end{align}
Let us study another case with $G_{11}=G_{21}=0$ and $G_{12},G_{22}\neq 0$. From Eq. \eqref{Eq:plus}, we obtain 
\begin{align}
k_{12}(-u)=k_{21}(-u)=0,\quad k_{11}(-u)=k_{22}(-u).
\end{align}
The corresponding local two-site state becomes

\begin{align}
\lvert\psi_0\rangle_{1,2}\propto\left(\begin{array}{cccc}
0&
-k_{11}(\frac{\eta}{2})&
k_{11}(\frac{\eta}{2})&
0
\end{array}\right)^{\mathrm{t}}.
\end{align}

\subsection{Construction of \texorpdfstring{$|{\Psi_{-,e}}\rangle$}{Psi\_{-,e}}}\label{sec:-e}

The compatibility relation \eqref{KTcompatII} leads to
\begin{align}
&\left(
\begin{array}{cc}
 G_{11} k_{22}(-u)-G_{21} k_{12}(-u) & G_{12} k_{22}(-u)-G_{22} k_{12}(-u) \\
 G_{21} k_{11}(-u)-G_{11} k_{21}(-u) & G_{22} k_{11}(-u)-G_{12} k_{21}(-u) \\
\end{array}
\right)\no\\
=&\left(
\begin{array}{cc}
 -G_{11} k_{22}(-u)+G_{12} k_{21}(-u) & -G_{12} k_{11}(-u)+G_{11} k_{12}(-u) \\
 -G_{21} k_{22}(-u)+G_{22} k_{21}(-u) & -G_{22} k_{11}(-u)+G_{21} k_{12}(-u) \\
\end{array}
\right).\label{Constraint:plus_2}
\end{align}
One can easily see that the above equation only has non-trivial solutions in the following cases: (i): $G_{11}=G_{22}=0$, $G_{12},G_{21}\ne 0$; (ii): $G=\sigma^z$. 
In the case $G_{11}=G_{22}=0$, we can get 
\begin{equation}
    k_{12}(-u)=-\frac{G_{12}}{G_{21}}k_{21}(-u),\quad 
    k_{11}(-u)=-k_{22}(-u),
\end{equation}
and the local two-site state then reads
\begin{equation}
    \lvert \psi_{0}\rangle_{1,2}\propto \left(\begin{array}{cccc}
k_{12}(\frac{\eta}{2})&
-k_{11}(\frac{\eta}{2})&
-k_{11}(\frac{\eta}{2})&
\frac{G_{21}}{G_{12}}k_{12}(\frac{\eta}{2})
\end{array}\right)^{t}.   
\end{equation}
Once $G=\sigma^z$, we obtain 
\begin{equation}
    k_{11}(-u)=k_{22}(-u)=0,
\end{equation}
and the local two-site state reads

\begin{equation}
    \lvert \psi_{0}\rangle_{1,2}\propto \left(\begin{array}{cccc}
k_{12}(\frac{\eta}{2})&
0&
0&
-k_{21}(\frac{\eta}{2})
\end{array}\right)^{t}.    
\end{equation}

An interesting observation is that once $G=\sigma^\alpha$, both $\lvert\Psi_{+,e}\rangle$ and $\lvert\Psi_{-,e}\rangle$ exist. This fact can be summarized in the following theorem.

\begin{theorem}\label{thm2}
Once $G=\sigma^\alpha,\,\,\alpha=x,y,z$, we can construct the following two-sites integrable boundary states with an even $L$
\begin{align}
|\Psi_{+,e}\rangle&=\left(\sum\limits_{\sigma=\pm}a_\sigma\,\lvert\varphi_{\alpha,\sigma}\rangle\otimes\lvert \varphi_{\alpha,-\sigma}\rangle\right)^{\otimes L/2},\\
|\Psi_{-,e}\rangle&=\left(\sum\limits_{\sigma=\pm}a_\sigma\,\lvert\varphi_{\alpha,\sigma}\rangle\otimes\lvert \varphi_{\alpha,\sigma}\rangle\right)^{\otimes L/2},
\end{align}
where $\sigma^{\alpha}\lvert\varphi_{\alpha,\pm}\rangle=\pm\lvert\varphi_{\alpha,\pm}\rangle$ and $a_\pm \in\mathbb{C}$.
\end{theorem}

\subsection{Construction of \texorpdfstring{$|{\Psi_{+,o}}\rangle$}{Psi\_{+,o}}}

In case of $G=\sigma^\alpha$, we can also construct $|{\Psi_{+,o}}\rangle$. 
The corresponding result can be summarized in the following theorem.
\begin{theorem}\label{thm3}
Once $G=\sigma^\alpha,\,\,\alpha=x,y,z$, we can construct the following factorized integrable boundary states with an odd $L$
\begin{align}
|\Psi_{+,o}\rangle&=\prod\limits_{{\rm even}\, n}\sigma_n^\alpha\,{\begin{pmatrix}
    1\\
    a
\end{pmatrix}}^{\otimes L},\quad a\in\mathbb{C}.
\end{align}
\end{theorem}
\begin{proof}
The proof of the Theorem  \ref{thm3} is similar to that demonstrated in Theorem \ref{thm1}. Let us consider the case $G=\sigma^z$ as an example.

Define 
\begin{align}
\mathcal{L}_{0,L}(u)\cdots \mathcal{L}_{0,2}(u)=\begin{pmatrix}\mathfrak{A}(u)&\mathfrak{B}(u)\\\mathfrak{C}(u)&\mathfrak{D}(u)\end{pmatrix},\quad 
\mathcal{L}_{0,1}(u)=\begin{pmatrix}\mathfrak{a}(u)&\mathfrak{b}(u)\\\mathfrak{c}(u)&\mathfrak{d}(u)\end{pmatrix}.
\end{align}
The transfer matrix can be expressed as
\begin{align}
t(u)=\mathfrak{A}(u)\mathfrak{a}(u)+\mathfrak{B}(u)\mathfrak{c}(u)-\mathfrak{C}(u)\mathfrak{b}(u)-\mathfrak{D}(u)\mathfrak{d}(u).
\end{align}
Acting the transfer matrix on the state $|\Psi_{+,o}\rangle$, one can derive 
\begin{align}
t(\pm u)|\Psi_{+,o}\rangle=\left(
\begin{array}{c}
\left[\frac{\sinh(\pm u+\frac{\eta}{2})}{\sinh\eta}\ \mathfrak{A}(\pm u)+a\mathfrak{B}(\pm u)-\frac{\sinh(\pm u-\frac{\eta}{2})}{\sinh\eta}\mathfrak{D}(\pm u)\right]\lvert\Phi_{0}\rangle\\ [6pt]
\left[(a\frac{\sinh(\pm u-\frac{\eta}{2})}{\sinh\eta}\mathfrak{A}(\pm u)-\mathfrak{C}(\pm u)-a\frac{\sinh(\pm u+\frac{\eta}{2})}{\sinh\eta}\mathfrak{D}(\pm u)\right]\lvert\Phi_{0}\rangle
\end{array}
\right),
\end{align}
where 
\begin{align}
\lvert\Phi_{0}\rangle=\begin{pmatrix}1\\-a\end{pmatrix}_2\otimes\begin{pmatrix}1\\a\end{pmatrix}_3\otimes\dots\otimes\begin{pmatrix}1\\-a\end{pmatrix}_{L-1}\otimes\begin{pmatrix}1\\a\end{pmatrix}_L.\label{Def:Psi0}
\end{align}
We see that the state $|{\Phi_0}\rangle$ defined in Eq. (\ref{Def:Psi0}) can be a two-site boundary state, which can be derived from Eq. (\ref{two:site:state}) with 
\begin{align}
\alpha_1=-\frac{\left(a^2+1\right) k_1(\frac{\eta }{2})}{2a k_4(\frac{\eta }{2})},\quad \alpha_2=\frac{\ii (a^2-1) k_2(\frac{\eta }{2})}{2ak_4(\frac{\eta}{2})},\quad \alpha_3=0\quad \alpha_4=1.
\label{Kmatsol}
\end{align}
The operators $\mathfrak{A}(u)$, $\mathfrak{B}(u)$, $\mathfrak{C}(u)$, $\mathfrak{D}(u)$ and the state $|{\Phi_0}\rangle$ now satisfy the same relations as in Eq. (\ref{TT:Relation}).
Using Eqs. (\ref{Kmatsol}) (\ref{TT:Relation}) and the following identities 
\begin{align}
&\frac{\sinh(u\pm\tfrac{\eta}{2})}{\sinh\eta}=\frac{1}{2}\left[\frac{k_4(u)}{k_4(\frac{\eta }{2})}\pm \frac{k_3(u)}{k_3(\frac{\eta }{2})}\right],\no\\
k_1(u)=k_2&(u)=1,\quad k_3(-u)=k_3(u),\quad 
k_4(-u)=-k_4(u),
\end{align}
we can thus prove
\begin{align}
t(u) \lvert\Psi_{+,o}\rangle = t(-u) \lvert\Psi_{+,o}\rangle.
\end{align}
Analogously, the theorem also holds for the case $G=\sigma^{x,y}$.
\end{proof}
It should be noted that $|{\Psi_{-,o}}\rangle$ should \textit{not} exist.
When $G=\sigma^z$, we see the eigenvalue of $t(u)$ has the following properties (see Eq. (\ref{TQsigz}))
\begin{align}
\lim_{u\to +\infty}\frac{\Lambda(u)}{(2\sinh\eta)^{-L}\ee^{L u}}&=2\sinh\left((\tfrac{L}{2}-M)\eta\right),\\
\lim_{u\to +\infty}\frac{\Lambda(-u)}{(2\sinh\eta)^{-L}\ee^{L u}}&=2(-1)^{L+1}\sinh\left((\tfrac{L}{2}-M)\eta\right),\quad M=0,1\ldots,L.
\end{align}
Hence, for an odd system size $L$, the identity $\Lambda(u) + \Lambda(-u) = 0$ cannot be satisfied. Consequently, no nontrivial quantum states can satisfy
$\langle \mathbf{u} |\,[t(u) + t(-u)]\,|\Psi_{-,o} \rangle = 0,$
which implies the absence of $|\Psi_{-,o}\rangle$.

For the case $G = \sigma^{x,y}$, we can also prove the non-existence of $|\Psi_{-,o}\rangle$. From Refs. \cite{KITANINE1999,Cao_2014,Wang2015}, we know that the eigenvalue of the transfer matrix $\Lambda(u)$ should satisfy
\begin{align}
\left[\Lambda(\tfrac{\eta}{2})\right]^L=\pm 1,\quad \Lambda(\tfrac{\eta}{2})\Lambda(-\tfrac{\eta}{2})=(-1)^{L+1}.\label{product:xxz}
\end{align}
When $L$ is odd, the condition $\Lambda(\eta/2) = -\Lambda(-\eta/2)$ contradicts Eq. \eqref{product:xxz}. 
This implies that the relation $\Lambda(u) = -\Lambda(-u)$ cannot be satisfied. 
Consequently, the boundary state $|\Psi_{-,o}\rangle$ does not exist for $G = \sigma^{x,y}$.

From Theorems \ref{thm2} and \ref{thm3}, we see that $G=\sigma^\alpha$ is more interesting than the generic case. Therefore, in the following, we focus on these specific cases. In the next sections, we let  $G=\sigma^z$ and $G=\sigma^x$ and derive the selection rules for Bethe roots corresponding to each type of boundary state.

\subsection{\label{sec:level1}Diagonal twist}
\paragraph{Integrability} We will start with a diagonal twist $G=\sigma^z$.
The Hamiltonian takes the form
\begin{equation}
H=\sum_{i=1}^{L}\sum_{\alpha=x,y,z}J^{\alpha}\sigma_i^{\alpha}\sigma_{i+1}^{\alpha},\quad  \sigma^{\alpha}_{L+1}=\sigma_1^z\sigma_1^{\alpha}\sigma_1^z. \label{Ham:xxz:z}
\end{equation}
The monodromy matrix in such case would be 
\begin{equation}
    \hat{T}_0(u)=G_0\mathcal{L}_{0,L}(u)\dots \mathcal{L}_{0,1}(u)=\begin{pmatrix}A(u)&B(u)\\-C(u)&-D(u)\end{pmatrix},
\end{equation}
where the $R$-matrix is given in Eq. (\ref{Rmatxxz}).
The transfer matrix is defined as usual by tracing out the auxiliary space $t(u)=\mathrm{Tr}_0(\hat{T}_0(u))$. The Hamiltonian in \eqref{Ham:xxz:z} can be given by the transfer matrix as follows
\begin{equation}
     H=2\sinh\eta\frac{\partial\ln t(u)}{\partial u}\Big\lvert_{u=\frac{\eta}{2}} - L\cosh\eta\times\mathbb{I}.
 \end{equation}

\paragraph{Exact solutions}
The twisted boundary condition does not break the 
$U(1)$ symmetry, so the eigenvalues and eigenvectors of the transfer matrix can be constructed via the standard algebraic Bethe ansatz.

The eigenvalue of the transfer matrix $\Lambda(u)$ can be parametrized by the  following $T\mbox{-}Q$ relation
\begin{equation}
    \Lambda(u)Q(u)=a(u)Q(u-\eta)-d(u)Q(u+\eta),
    \label{TQsigz}
\end{equation}
where $a(u)$ and $d(u)$ are defined in Eq. \eqref{def:ab} and $Q(u)$ reads \cite{He_2025}
\begin{align}
Q(u)=\prod_{j=1}^M\frac{\sinh(u-u_j)}{\sinh\eta},\quad M=0,1,\ldots,L.
\end{align}
The Bethe roots $\{u_1,\ldots,u_M\}$ should satisfy the Bethe equations 
\begin{equation}
    -\left[\frac{\sinh(u_j+\frac{\eta}{2})}{\sinh(u_j-\frac{\eta}{2})}\right]^L=\prod_{k\neq j}^M\frac{\sinh(u_j-u_k+\eta)}{\sinh(u_j-u_k-\eta)},\quad j=1,\dots,M.
\end{equation}
The eigenvectors of $t(u)$ can be parameterized as  
\begin{equation}
\begin{split}
    \lvert\mathbf{u}\rangle=B(u_1)\dots B(u_M)\lvert0\rangle,\\
    \langle \mathbf{u}|=\langle 0|C(u_1)\dots C(u_M).
\end{split}
\end{equation}
The corresponding energy in terms of the Bethe roots is  
\begin{equation}
    E(\mathbf{u})=2\sum_{j=1}^{M}[g(u_j-\tfrac{\eta}{2})-g(u_j+\tfrac{\eta}{2})]+Lg(\eta),
\end{equation}
which is the same as the periodic case. The norm of the above eigenvectors take the same form as the case of periodic chain, which is Eq. (\ref{norm_peri}). 
\paragraph{Selection rule}
The overlap between $|{\Psi_{\pm,s}}\rangle$ and  $|{\mathbf{u}}\rangle$ in non-zero only when $\Lambda(u)=\pm \Lambda(-u)$, which will give specific selection rule for $\{u_1,\ldots,u_M\}$. For the states $|\Psi_{-,e}\rangle$ and $|\Psi_{+,o}\rangle$, the Bethe roots must be paired as follows
\begin{equation}
\{u_1,\dots,u_M\}=\{-u_1+\ii\pi l_1,\dots,-u_M+\ii\pi l_M\},\quad l_i\in\mathbb{Z}. \label{Selection:DBC1}
\end{equation}
For the state $|\Psi_{+,e}\rangle$, one finds that the Bethe roots do not exhibit a paired structure. However, the condition $\Lambda(u)=\Lambda(-u)$ still imposes multiple constraints on the Bethe roots. To analyze them, we examine the properties of $\Lambda(u)$ at the points $u=\pm\eta/2$. And this yields the following relations
\begin{align}
\begin{aligned}\label{constraint:z:plus}
 \frac{d^{k}\Lambda(u)}{du^{k}}\Big\lvert_{u=\frac{\eta}{2}}&=\frac{d^{k}\Lambda(u)}{du^{k}}\Big\lvert_{u=-\frac{\eta}{2}},\quad k=0,2,\ldots,\\
 \frac{d^{k}\Lambda(u)}{du^{k}}\Big\lvert_{u=\frac{\eta}{2}}&=-\frac{d^{k}\Lambda(u)}{du^{k}}\Big\lvert_{u=-\frac{\eta}{2}},\quad k=1,3,\ldots.
 \end{aligned}
\end{align}
By considering the first three terms ($k=0,1,2$), we get the following selection rule 
\begin{align}
\begin{aligned}\label{selection:rule:z}
&\prod_{j=1}^M\frac{\sinh(u_j+\frac{\eta}{2})}{\sinh(u_j-\frac{\eta}{2})}=\pm \ii^{L-1},\\
&\sum_{j=1}^M\frac{1}{\sinh^2(u_j+\frac{\eta}{2})}=\sum_{j=1}^M\frac{1}{\sinh^2(u_j-\frac{\eta}{2})}.
\end{aligned}
\end{align}
It should be noted that Eq. \eqref{selection:rule:z} is a necessary condition for $\Lambda(u)=\Lambda(-u)$. Considering a larger $k$ in Eq. \eqref{constraint:z:plus} or analyzing $\Lambda(u)$ at other points may yield further constraints on the Bethe roots. 

From Eq. (\ref{TQsigz}), we can derive
\begin{align}
\begin{aligned}
&\lim_{u\to +\infty}\frac{\Lambda(u)}{(2\sinh\eta)^{-L}\ee^{L u}}-
\lim_{u\to +\infty}\frac{\Lambda(-u)}{(2\sinh\eta)^{-L}\ee^{L u}}=4\sinh\left((\tfrac{L}{2}-M)\eta\right).
\end{aligned}
\end{align}
It is obvious that the overlap $\langle \mathbf{u} | \Psi_{+,e} \rangle$ vanishes except in the sector with $M = L/2$. In other words, $M=L/2$ is a supplementary condition for the selection rule shown in \eqref{selection:rule:z}.
\paragraph{Boundary state overlaps}
The states $\lvert \Psi_{-,e}\rangle$ and $\lvert\Psi_{+,o}\rangle$ have paired selection rules from the above analysis. 
Choosing the following parameterization of the boundary $K$-matrix
\begin{equation}
\begin{split}
    \alpha_1=\cosh \theta,&\quad 
    \alpha_2=i \sinh \theta,\\
    \alpha_3=2\cosh \alpha \sinh\beta,&\quad \alpha_4=2\sinh\alpha\cosh\beta,
\end{split}    
\end{equation}
then the overlap between integrable boundary states $\lvert \Psi_{-,e}\rangle$ and Bethe states take the following
\begin{equation}
   \frac{\langle\mathbf{u}\lvert \Psi_{-,e}\rangle}{\sqrt{\langle\mathbf{u}\lvert\mathbf{u}\rangle}}=|\ee^{\theta(L-2M)}|^{1/2}(\sinh\eta)^{L/2}\Big(\prod_{j=1}^{M/2}y_l(\alpha,\beta,u_j)\Big)^{1/2}\sqrt{\frac{\det G^+}{\det G^-}}
\end{equation}
with the definition (\ref{gaulike_det}) and (\ref{yl_def}).
For the case of $\lvert\Psi_{+,o}\rangle$, one can change the parameter to 
\begin{equation}
    \alpha=i\pi/2,\quad \beta=\eta/2+i\pi/2
\end{equation}
with $a$ as arbitrary parameter, then the overlap formula would be 
\begin{equation}
    \frac{\langle\mathbf{u}\lvert \Psi_{+,o}\rangle}{\sqrt{\langle\mathbf{u}\lvert\mathbf{u}\rangle}}=|a^{2M-L}|^{1/2}(-a)^{L/2}\Big(\prod_{j=1}^{M/2}y_l(\alpha,\beta,u_j)\Big)^{1/2}\sqrt{\frac{\det G^+}{\det G^-}}.
\end{equation}


\subsection{Off-diagonal twist}
\paragraph{Integrability}
For an XXZ chain, the twists $G=\sigma^x$ and $G=\sigma^y$ play identical roles. In the following, we focus on the case $G=\sigma^x$. 
The Hamiltonian of the XXZ chain now reads
\begin{equation}
H=\sum_{i=1}^L\sum_{\alpha=x,y,z}J^\alpha\sigma^{\alpha}_i\sigma^{\alpha}_{i+1},\quad \sigma^{\alpha}_{L+1}\equiv\sigma_1^x\sigma_1^{\alpha}\sigma_1^x.\label{Ham:twsit:x}
\end{equation}
The monodromy matrix takes the form
\begin{equation}
    \tilde{T}_0(u)=\sigma_0^x\mathcal{L}_{0,L}(u)\dots \mathcal{L}_{0,1}(u)=\begin{pmatrix}C(u)&D(u)\\A(u)&B(u)\end{pmatrix},
\end{equation}
and the transfer matrix $t(u)$ is defined as 
\begin{equation}
    t(u)=\mathrm{Tr}_0(\tilde{T}_0(u))=B(u)+C(u).
\end{equation}
The first derivative of $\log t(u)$ gives the Hamiltonian in \eqref{Ham:twsit:x}
\begin{equation}
     H=2\sinh\eta\frac{\partial\ln t(u)}{\partial u}\Big\lvert_{u=\frac{\eta}{2}} - L\cosh\eta\times\mathbb{I}.
 \end{equation}

\paragraph{Exact solutions}
The topological boundary condition $\sigma^{\alpha}_{L+1}\equiv\sigma_1^x\sigma_1^{\alpha}\sigma_1^x$ breaks the $U(1)$-symmetry. In this case, the $T\mbox{-}Q$ relation can also be constructed with an inhomogeneous term \cite{Cao_2013,Wang2015}
\begin{align}
\Lambda(u)\tilde{Q}(u)=\ee^{u-\frac{\eta}{2}}a(u) \tilde{Q}(u-\eta)-\ee^{-u-\frac{\eta}{2}}d(u)\tilde{Q}(u+\eta)-c(u)a(u)d(u),
\label{ofTQ}
\end{align}
where $a(u)$ and $d(u)$ are given by Eq. \eqref{def:ab}, the function $\tilde Q(u)$ in \eqref{ofTQ} is a trigonometric polynomial in $u$ of fixed degree $L$ \cite{jiang2025rationalqsystemsintegrablespin}
\begin{align}
\tilde{Q}(u)&=\prod_{j=1}^L\frac{\sinh(u-u_j)}{\sinh\eta},
\end{align}
and $c(u)$ is defined as
\begin{equation}
\begin{split}
c(u)&=\exp\left[u-\frac{\eta}{2}-\sum_{l=1}^L\left(u_l+\frac{\eta}{2}\right)\right]-\exp\left[-u-\frac{\eta}{2}+\sum_{l=1}^L\left(u_l-\frac{\eta}{2}\right)\right].
\end{split}
\end{equation}
The corresponding BAEs now take the following form
\begin{equation}
\ee^{u_j-\frac{\eta}{2}}a(u_j)\tilde{Q}(u_j-\eta)-\ee^{-u_j-\frac{\eta}{2}}d(u_j)\tilde{Q}(u_j+\eta)-c(u_j)a(u_j)d(u_j)=0,\quad j=1,\dots ,L.
    \label{BAE-xxz-odia}
\end{equation}
The energy of the model now is
\begin{equation}
     E(\mathbf{u})=2\sum_{j=1}^{L}[g(u_j-\tfrac{\eta}{2})-g(u_j+\tfrac{\eta}{2})]+Lg(\eta)+2\sinh\eta.
\end{equation}
Based on the inhomogeneous $T$-$Q$ relation \eqref{ofTQ},
the Bethe state has been retrieved \cite{Zhang_2015}, specifically as follows 
\begin{equation}
\lvert\mathbf{u}\rangle=\prod_{i=1}^LD(u_i)\lvert\Omega\rangle,\label{beth-xxz-od}
\end{equation}
where $\{u_i\lvert i=1,\dots,L\}$ are given by Eq. (\ref{BAE-xxz-odia}). The reference state $|{\Omega}\rangle$ in \eqref{beth-xxz-od} is a highly entangled state
\begin{equation}
    \lvert\Omega\rangle=\sum_{l=0}^L\frac{(B^-)^l}{[l]_q!}\lvert0\rangle,
\end{equation}
where 
\begin{align}
[l]_q&=\frac{1-q^{2l}}{1-q^2},\quad  [l]_q!=[l]_q[l-1]_{q}\dots[1]_q,\quad q=\ee^{\eta},\\
B^-
&=\sum_{l=1}^L\ee^{\frac{L-1}{2}\eta}\ee^{\frac{\eta}{2}\sum_{k={l+1}}^L\sigma_k^z}\sigma_l^-\ee^{-\frac{\eta}{2}\sum_{k=1}^{l-1}\sigma_k^z}.
\end{align}

\paragraph{Selection rule}
To make the overlap $\langle{\mathbf{u}}|{\Psi_{\pm,s}}\rangle$ non-zero, we need either $\Lambda(u)=\Lambda(-u)$ or $\Lambda(u)=-\Lambda(-u)$, each of which imposes a specific selection rule on the Bethe roots. For the states $|\Psi_{-,e}\rangle$ and $|\Psi_{+,o}\rangle$, the selection rules are clear and the Bethe roots should form the following pattern $$\{u_1,\dots,u_L\}=\{-u_1+\ii\pi l_1,\dots,-u_L+\ii\pi l_L\},\quad l_i\in\mathbb{Z}.$$
In contrast, for the state $|\Psi_{+,e}\rangle$, no simple root pattern exists. Using the method illustrated in Eq. \eqref{constraint:z:plus}, we can get the following selection rule by analyzing the properties of $\Lambda(u)$ at $u=\pm\eta/2$
\begin{align}
\begin{aligned}\label{selection:rule:x1}
\prod_{j=1}^L\frac{\sinh(u_j+\frac{\eta}{2})}{\sinh(u_j-\frac{\eta}{2})}=\pm \ii^{L-1},\qquad 
\sum_{j=1}^L\frac{1}{\sinh^2(u_j+\frac{\eta}{2})}=\sum_{j=1}^L\frac{1}{\sinh^2(u_j-\frac{\eta}{2})}.
\end{aligned}
\end{align}

\section{Integrable boundary states in the XYZ chain}
\label{sec4}
In this section, we begin by reviewing the integrability of $\mathrm{XYZ}$ chain under a generic twist. 
Then, we construct the two-site integrable boundary state for an even number of sites $L$ and extend the construction of boundary states to odd $L$ with a rigorous proof. 
In the end, we analyze the $T$-$Q$ relation and derive the associated selection rules.
\subsection{Integrability of the XYZ chain}
\label{sec4.1}
Let us first introduce the following functions
\begin{equation}
    w_i(u)=\frac{\theta_{5-i}(u+\tfrac{\eta}{2})}{2\theta_{5-i}(\tfrac{\eta}{2})},\quad i=1,\dots,4,
\end{equation}
where $\theta_i(u)$ are theta-functions (see Appendix \ref{Appendix_A}). 
The $R$-matrix corresponding to the XYZ model is \cite{baxter1982,Slavnov_2020}
\begin{equation}
    R_{0,j}(u)=w_1(u)\sigma_0^x\otimes\sigma_j^x+w_2(u)\sigma_0^y\otimes\sigma_j^y+w_3(u)\sigma_0^z\otimes\sigma_j^z+w_4(u)\mathbb{I}_0\otimes\mathbb{I}_j.
\end{equation}
The Lax operator $\mathcal{L}_{0,j}$ is related to the $R$-matrix by a shift
\begin{equation}
    \mathcal{L}_{0,j}(u)=R_{0,j}(u-\tfrac{\eta}{2}),\quad \mathrm{with}\quad \mathcal{L}_{0,j}(u)=-\sigma_0^y\mathcal{L}^{t_0}_{0,j}(-u)\sigma_1^y,
\end{equation}
and satisfies the RLL relation 
\begin{equation}
    R_{0,\bar{0}}(u-v)\mathcal{L}_{0,j}(u)\mathcal{L}_{\bar{0},j}(v)=\mathcal{L}_{\bar{0},j}(v)\mathcal{L}_{0,j}(u)R_{0,\bar{0}}(u-v).\label{RLL:xyz}
\end{equation}
The monodromy matrix is defined as 
\begin{equation}
    T_0(u)=G_0\mathcal{L}_{0,L}(u)\dots \mathcal{L}_{0,1}(u),
\end{equation}
where $G$ represents the twist. To guarantee the integrability, the following 
$\mathrm{RTT}$-relation needs to be satisfied 
\begin{equation}
R_{0,\bar{0}}(u-v)T_0(u)T_{\bar{0}}(v)=T_0(u)T_{\bar{0}}(v)R_{0,\bar{0}}(u-v).\label{RTT:xyz}
\end{equation}
This implies the commutation relation $R_{0,\bar{0}}(u) G_0 G_{\bar{0}} = G_0 G_{\bar{0}} R_{0,\bar{0}}(u)$, which yields the following solutions for the twist operator $G$
\begin{align}
G=\mathbb{I},\quad \mbox{or}\quad G=\sigma^{x,y,z}.
\end{align}
The corresponding transfer matrix is 
\begin{align}
t(u)=\mathrm{Tr}_0(T_0(u)).\label{t:xyz}
\end{align}
With the help of Eq. \eqref{RTT:xyz}, we can prove $[t(u),\,t(v)]=0$. Similar to the XXZ case, one can obtain the Hamiltonian $H$ by taking the logarithmic derivative of the transfer matrix
\begin{align}
H&=2\frac{\ell{1}(\eta)}{\ell{1}'(0)}\frac{\partial\ln t(u)}{\partial u}\Big\lvert_{u=\frac{\eta}{2}} - L\frac{\ell{1}'(\eta)}{\ell{1}'(0)}\times\mathbb{I}\no\\
&=\sum_{i=1}^L\,\sum_{\alpha=x,y,z}J^{\alpha}\sigma_i^{\alpha}\sigma_{i+1}^{\alpha},\quad  \label{Ham:xyz}
\end{align}
where the boundary condition is 
$\sigma_{L+1}^\alpha\equiv G_1\sigma_1^\alpha G_1$
and the coupling constants in Eq. \eqref{Ham:xyz} are parametrized as 
\begin{equation}
    J_x=\frac{\theta_4(\eta)}{\theta_4(0)},\quad J_y=\frac{\theta_3(\eta)}{\theta_3(0)},\quad J_z=\frac{\theta_2(\eta)}{\theta_2(0)}.
\end{equation}

\subsection{Construction of \texorpdfstring{$|{\Psi_{\pm,e}}\rangle$}{Psi\_{pm,e}}}\label{sec:level2} 
The notion of generalized integrable boundary states extends directly from the XXZ to the XYZ spin chain. We begin by recalling their definition
\begin{align}
t(u) |\Psi_{\pm, s}\rangle = \pm t(-u) |\Psi_{\pm, s}\rangle, \quad s = e, o,
\end{align}
where $e$ and $o$ denote systems with an even and odd total number of sites, respectively. 

The integrable boundary states $|\Psi_{\pm,e}\rangle$ for an even $L$ can be obtained directly form the $K$-matrix, which has been discussed in Ref. \cite{Yuan_2024}. 
The generic $K$-matrix corresponding to the XYZ chain is \cite{BOYU_1993,Inami_1994,Guan1996}
\begin{align}
K(u)=\frac{\alpha_1}{k_1(u)}\sigma^x+\frac{\alpha_2}{k_2(u)}\sigma^y+\frac{\alpha_3}{k_3(u)}\sigma^z+\frac{\alpha_4}{k_4(u)}\,\mathbb{I},\quad k_i(u)=\theta_{5-i}(u),
\label{xyz-Kmat}
\end{align}
where $\alpha_1,\alpha_2,\alpha_3$ and $\alpha_4$ are free parameters. It obeys the following boundary Yang-Baxter equation
\begin{equation}
R_{1,2}(u-w)K_1(u)R_{2,1}(u+w)K_2(w)=K_2(w)R_{1,2}(u+w)K_1(u)R_{2,1}(u-w).
\end{equation}
Following Theorem \ref{thm0}, we can build the two-site states from the $K$-matrix as follows \cite{Yuan_2024}
\begin{equation}
    \lvert \Psi_{+,e}\rangle=\lvert\psi_0\rangle_{1,2}^{\otimes L/2},\qquad \lvert\psi_0\rangle_{1,2}=\sum_{i,j=1}^2[K(\tfrac{\eta}{2})\sigma^y]_{ij}\lvert i\rangle_1\otimes\lvert j\rangle_2.
    \label{xyz-two-site}
\end{equation}

In the presence of twists $G=\sigma^a$, we find that Theorem \ref{thm2} is also applicable to the XYZ chain. The construction of $|\Psi_{\pm,e}\rangle$ follows the same procedure as in Sections~\ref{sec:+e} and \ref{sec:-e}.
We can solve the compatibility conditions in Eqs. (\ref{KTcompatI}) or (\ref{KTcompatII}) to get proper integrable boundary states $|{\Psi_{\pm,e}}\rangle$, and the integrability condition of the states are given by the corresponding $\mathrm{KT}$-relation respectively.

\subsection{Construction of \texorpdfstring{$|{\Psi_{\pm,o}}\rangle$}{Psi\_{pm,o}}}

It is also possible to construct integrable boundary states for the XYZ chain on odd-length lattices. We find Theorems~\ref{thm1} and~\ref{thm3} all remain valid for the XYZ chain. As an example, we demonstrate the existence of the following boundary state
\begin{equation}
\lvert\Psi_{+,o}\rangle=\begin{pmatrix}1\\a\end{pmatrix}^{\otimes L}
,\quad a\in \mathbb{C},\quad \mbox{with}\,\, G=\mathbb{I}.\label{Boundary:State:Twist_xyz}
\end{equation}
\begin{proof}
Let us define 
\begin{align}
\mathcal{L}_{0,L}(u)\cdots \mathcal{L}_{0,2}(u)=\begin{pmatrix}\mathfrak{A}(u)&\mathfrak{B}(u)\\\mathfrak{C}(u)&\mathfrak{D}(u)\end{pmatrix},\quad 
\mathcal{L}_{0,1}(u)=\begin{pmatrix}\mathfrak{a}(u)&\mathfrak{b}(u)\\\mathfrak{c}(u)&\mathfrak{d}(u)\end{pmatrix}.
\end{align}
The transfer matrix defined in Eq. \eqref{t:xyz} can be expressed as
\begin{align}
t(u)=\mathfrak{A}(u)\mathfrak{a}(u)+\mathfrak{B}(u)\mathfrak{c}(u)+\mathfrak{C}(u)\mathfrak{b}(u)+\mathfrak{D}(u)\mathfrak{d}(u).
\end{align}
Acting the transfer matrix on the state $|\Psi_{+,o}\rangle$, one can derive 
\begin{align}
2t(u)|\Psi_{-,o}\rangle=&\binom{(\bell{1}(u)+\bell{2}(u))\mathfrak{A}(u)\ket{\Phi_0}}{a(\bell{1}(u)-\bell{2}(u))\mathfrak{A}(u)\ket{\Phi_0}}+\binom{(\bell{1}(u)-\bell{2}(u))\mathfrak{D}(u)\ket{\Phi_0}}{a(\bell{1}(u)+\bell{2}(u))\mathfrak{D}(u)\ket{\Phi_0}}\no\\
&+\binom{a(\bell{4}(u)+\bell{3}(u))\mathfrak{B}(u)\ket{\Phi_0}}{(\bell{4}(u)-\bell{3}(u))\mathfrak{B}(u)\ket{\Phi_0}}+\binom{a(\bell{4}(u)-\bell{3}(u))\mathfrak{C}(u)\ket{\Phi_0}}{(\bell{4}(u)+\bell{3}(u))\mathfrak{C}(u)\ket{\Phi_0}},
\end{align}
where 
\begin{align}
\lvert\Phi_0\rangle=\begin{pmatrix}1\\a\end{pmatrix}_2\otimes\begin{pmatrix}1\\a\end{pmatrix}_3\otimes\dots\otimes\begin{pmatrix}1\\a\end{pmatrix}_{L-1}\otimes\begin{pmatrix}1\\a\end{pmatrix}_L.\label{Def:Psi0xyz}
\end{align}
And here we use the short notation 
\begin{align}
\tilde{\theta}_i(u)=\frac{\theta_i(u)}{\theta_i(\frac{\eta}{2})}=\frac{k_{5-i}(u)}{k_{5-i}(\frac{\eta}{2})},\quad i=1,2,3,4.
\end{align}
The state defined in Eq. (\ref{Def:Psi0xyz}) can be seen as a two-sites boundary state, which can be derived from Eq. (\ref{xyz-two-site}) with 
\begin{align}
\alpha_1=-\frac{\left(a^2+1\right) k_1(\frac{\eta }{2})}{2a k_4(\frac{\eta }{2})},\quad \alpha_2=\frac{\ii (a^2-1) k_2(\frac{\eta }{2})}{2ak_4(\frac{\eta}{2})},\quad \alpha_3=0,\quad \alpha_4=1.
\end{align}
The matrices $\mathfrak{A}(u)$, $\mathfrak{B}(u)$, $\mathfrak{C}(u)$, $\mathfrak{D}(u)$ and the state $|\Phi_0\rangle$ now satisfy the same relations as in Eq. (\ref{TT:Relation}).
Equations (\ref{two:site:state}), (\ref{TT:Relation}) together with the following identities 
\begin{align}
k_1(-u)=k_1(u),\quad k_2(-u)=k_2(u),\quad k_3(-u)=k_3(u),\quad 
k_4(-u)=-k_4(u),
\end{align}
give rise to
\begin{align}
t(u) \lvert\Psi_{-,o}\rangle = t(-u) \lvert\Psi_{-,o}\rangle.
\end{align}
\end{proof}
For other integrable boundary states $|\Psi_{\pm,o}\rangle$,  the same proof technique applies.  We therefore omit the details.

\subsection{Non‑existence of specific integrable boundary states}\label{sec:non-BS}
From Refs. \cite{KITANINE1999,Cao_2014,Wang2015}, we know that the eigenvalue of the transfer matrix $\Lambda(u)$ should satisfy
\begin{align}
&G=\mathbb{I}:\qquad \left[\Lambda(\tfrac{\eta}{2})\right]^L=1,\quad \Lambda(\tfrac{\eta}{2})\Lambda(-\tfrac{\eta}{2})=(-1)^{L},\label{product:xyz:1}\\
&G=\sigma^\alpha:\qquad \left[\Lambda(\tfrac{\eta}{2})\right]^L=\pm1,\quad \Lambda(\tfrac{\eta}{2})\Lambda(-\tfrac{\eta}{2})=(-1)^{L+1}.\label{product:xyz:2}
\end{align}
With the help of Eqs. \eqref{product:xyz:1} and~\eqref{product:xyz:2}, we can conclude that the condition $\Lambda(\frac{\eta}{2}) = \pm \Lambda(-\frac{\eta}{2})$ cannot be satisfied in the following cases: (1) $|\Psi_{+,o}\rangle$ for $G = \mathbb{I}$; (2) $|\Psi_{-,o}\rangle$ for $G = \sigma^\alpha$ ($\alpha = x, y,z$). Consequently, the corresponding boundary states are absent.  

In case of $G = \mathbb{I}$, the identity $\Lambda(\eta/2) = -\Lambda(-\eta/2)$ can be satisfied only when $L = 4l,\,l \in \mathbb{N}^+$. From Eq. \eqref{Leading:Term:PBC}, we know that the non-trivial $\lvert \Psi_{-,e} \rangle$ cannot be constructed in the periodic $\mathrm{XXZ}$ spin chain. Since the XXZ chain can be regarded as a specific smooth limit of the $\mathrm{XYZ}$ chain, it is natural to infer that $\lvert \Psi_{-,e} \rangle$ does not exist in the periodic XYZ model either.

\section{Exact solutions of the XYZ chain and selection rules}
\label{sec5}

\subsection{Periodic XYZ  chain}
Due to the lack of ${U(1)}$ symmetry in the $\mathrm{XYZ}$ chain, its reference state is rather complicated. Therefore, we will not delve into the details of constructing Bethe eigenstates \cite{BAXTER1973_1, Slavnov_2020}.
From Baxter's pioneering work \cite{PhysRevLett.26.832,baxter1982}, the $T\mbox{-}Q$ relation for the periodic XYZ chain with an even $L$ is given by \cite{baxter1982}
\begin{align}
\Lambda(u)=&\,\ee^{\ii\pi\eta\nu}\left[\frac{\theta_1(u+\frac{\eta}{2})}{\theta_1(\eta)}\right]^L\prod_{k=1}^M\frac{\theta_1(u-u_k-\eta)}{\theta_1(u-u_k)}\no\\
&\,+\ee^{-\ii\pi\eta\nu}\left[\frac{\theta_1(u-\frac{\eta}{2})}{\theta_1(\eta)}\right]^L\prod_{k=1}^M\frac{\theta_1(u-u_k+\eta)}{\theta_1(u-u_k)},
\end{align}
where $M=L/2$ and the parameter $\nu$ is an integer. The Bethe roots $\{u_k\}$ satisfy the following Bethe equations
\begin{align}
&\ee^{2\ii\pi\eta\nu}\,\left[\frac{\theta_{1}(u_j+\frac{\eta}{2})}{\theta_{1}(u_j-\frac{\eta}{2})}\right]^{L}\,\prod_{k\neq j}^{M}\frac{\theta_{1}(u_j-u_k-\eta)}{\theta_{1}(u_j-u_k+\eta)}=1,\quad j=1,2,\ldots,M,\label{BAE;M;1}
\end{align}
and the valid Bethe roots satisfy the selection rule
\begin{align}
2\sum_{j=1}^{M}u_j=k+\nu\tau,\qquad k,\nu\in\mathbb{Z}.\label{SelectionRule}
\end{align}  
The energy of Hamiltonian in terms of Bethe roots is 
\begin{align}
&E(\mathbf{u})=2\sum_{j=1}^{M}[f(u_j-\tfrac{\eta}{2})-f(u_j+\tfrac{\eta}{2})]+Lf(\eta),\quad f(u)=\frac{\theta_{1}(\eta)\theta_{1}'(u)}{\theta_{1}'(0)\theta_{1}(u)}.\label{Energyxyzp}
\end{align}
\paragraph{Selection rule given by $|\Psi_{+,e}\rangle$}
To make the overlap between integrable boundary state $\lvert \Psi_{+,e}\rangle$ and Bethe eigen-states non-zero, we need the constraint $\Lambda(u)=\Lambda(-u)$, which gives us the following selection rule for the Bethe roots configuration ($M=L/2$)
\begin{equation}
\{u_1,\dots,u_M\}=\{-u_1+l_1\tau+j_1,\dots,-u_M+l_M\tau+j_M\},\quad \nu=\sum_{i=1}^Ml_i,\quad l_i,j_i\in\mathbb{Z}.\label{Selection:PBCOxyz}
\end{equation}

Once $L$ is odd, it has been discussed in Ref. \cite{Wang2015} that the $T$-$Q$ relation should contain an inhomogeneous term and can take the following form 
\begin{align}
\Lambda(u)=&\,\ee^{\ii\pi l_1 (u-\frac{\eta}{2})+\ii\xi}\left[\frac{\theta_1(u+\frac{\eta}{2})}{\theta_1(\eta)}\right]^L\frac{Q_1(u-\eta)}{Q_2(u)}+\ee^{-\ii\pi l_1 (u+\frac{\eta}{2})-\ii\xi}\left[\frac{\theta_1(u-\frac{\eta}{2})}{\theta_1(\eta)}\right]^L\frac{Q_2(u+\eta)}{Q_1(u)}\no\\
&+c\frac{\ell{1}(u)}{\ell{1}(\eta)}\left[\frac{\theta_1(u+\frac{\eta}{2})}{\theta_1(\eta)}\right]^L\left[\frac{\theta_1(u-\frac{\eta}{2})}{\theta_1(\eta)}\right]^L\frac{1}{Q_1(u)Q_2(u)},
\end{align}
where $l_1$ is an \textit{even} integer and 
\begin{align}
Q_1(u)=\prod_{j=1}^{M}\frac{\ell{1}(u-\mu_j)}{\ell{1}(\eta)},\quad Q_2(u)=\prod_{j=1}^{M}\frac{\ell{1}(u-\nu_j)}{\ell{1}(\eta)},\quad M=\frac{L+1}{2}.
\end{align}
The $2M+2$ parameters $\xi$, $c$, $\{\mu_j\}$ and $\{\nu_j\}$ should satisfy the following Bethe equations 
\begin{align}
\begin{aligned}\label{BAE:xyz:odd}
&(L-2M)\eta-2\sum_{j=1}^M(\mu_j-\nu_j)=l_1\tau+2m_1,\quad m_1\in \mathbb{Z},\\
&\sum_{j=1}^{M}(\mu_j+\nu_j)=m_2,\quad m_2\in \mathbb{Z},\\
&c\frac{\ee^{\ii\pi l_1(\mu_j+\frac{\eta}{2})+\ii\xi}\ell{1}(\mu_j)}{\ell{1}(\eta)}\left[\frac{\theta_1(\mu_j+\frac{\eta}{2})}{\theta_1(\eta)}\right]^L=-Q_2(\mu_j)Q_2(\mu_j+\eta),\\
&c\frac{\ee^{-\ii\pi l_1(\nu_j-\frac{\eta}{2})-\ii\xi}\ell{1}(\nu_j)}{\ell{1}(\eta)}\left[\frac{\theta_1(\nu_j-\frac{\eta}{2})}{\theta_1(\eta)}\right]^L=-Q_1(\nu_j)Q_1(\nu_j-\eta),\\
&\ee^{\ii L\xi}\left[\frac{Q_1(-\frac{\eta}{2})}{Q_2(\frac{\eta}{2})}\right]^L=1.
\end{aligned}
\end{align}
The energy of the system can be given by the Bethe roots as follows
\begin{align}
E(\mathbf{\mu},\mathbf{\nu})=2\sum_{j=1}^M[f(\nu_j-\tfrac{\eta}{2})-f(\mu_j+\tfrac{\eta}{2})]+Lf(\eta)+2\ii\pi l_1\frac{\ell{1}(\eta)}{\ell{1}'(0)},
\end{align}
where $f(u)$ is defined in Eq. \eqref{Energyxyzp}.

\paragraph{Selection rule given by $|\Psi_{-,o}\rangle$} The boundary state $|\Psi_{-,o}\rangle$ can select the Bethe state $|\Psi(\mathbf{\mu},\mathbf{\nu})\rangle$ satisfying $\Lambda(u)|\Psi(\mathbf{\mu},\mathbf{\nu})\rangle=-\Lambda(-u)|\Psi(\mathbf{\mu},\mathbf{\nu})\rangle$. For the overlap $\langle\Psi(\mathbf{\mu},\mathbf{\nu})|\Psi_{-,o}\rangle$ to be non-zero, the Bethe roots and the parameter $\xi$ should satisfy the following constraint
\begin{align}
 \{\nu_1,\ldots,\nu_M\}=\{-\mu_1,\ldots,-\mu_M\},\quad \ee^{\ii\xi}=(-1)^{\frac{L(L+1)}{2}}.
\end{align}


\subsection{XYZ chain with twisted boundary}
In the presence of a $\sigma^x$ twist, the construction of transfer matrix is laid out in Section \ref{sec4.1}, while the corresponding Hamiltonian reads (\ref{Ham:xyz}). The construction of Bethe states is not well-established in this case, therefore, we will not delve too deeply into this matter. The eigenvalue of the transfer matrix can be parameterized by the following inhomogeneous $T$-$Q$ relation \cite{Wang2015}
\begin{align}
\Lambda(u)=&\,\ee^{\ii\pi l_1 (u-\frac{\eta}{2})+\ii\xi}\left[\frac{\theta_1(u+\frac{\eta}{2})}{\theta_1(\eta)}\right]^L\frac{Q_1(u-\eta)}{Q_2(u)}-\ee^{-\ii \pi l_1(u+\frac{\eta}{2})-\ii\xi}\left[\frac{\theta_1(u-\frac{\eta}{2})}{\theta_1(\eta)}\right]^L\frac{Q_2(u+\eta)}{Q_1(u)}\no\\
&+c\frac{\ee^{\ii \pi u}\ell{1}^m(u)}{\ell{1}^m(\eta)}\left[\frac{\theta_1(u+\frac{\eta}{2})}{\theta_1(\eta)}\right]^L\left[\frac{\theta_1(u-\frac{\eta}{2})}{\theta_1(\eta)}\right]^L\frac{1}{Q_1(u)Q_2(u)},
\end{align}
where $l_1$ is an \textit{odd} integer, $m$ is 0 for even $L$ and 1 for odd $L$ and 
\begin{align}
Q_1(u)=\prod_{j=1}^{M}\frac{\ell{1}(u-\mu_j)}{\ell{1}(\eta)},\quad Q_2(u)=\prod_{j=1}^{M}\frac{\ell{1}(u-\nu_j)}{\ell{1}(\eta)},\quad M=\frac{L+m}{2}.
\end{align}
The $2M+2$ parameters $\xi$, $c$, $\{\mu_j\}$ and $\{\nu_j\}$ should satisfy the following Bethe equations 
\begin{align}
\begin{aligned}
&(L-2M)\eta-2\sum_{j=1}^M(\mu_j-\nu_j)=l_1\tau+2m_1,\quad m_1\in \mathbb{Z},\\
&2\sum_{j=1}^{M}(\mu_j+\nu_j)=2m_2+\tau,\quad m_2\in \mathbb{Z},\\
&c\frac{\ee^{\ii\pi \mu_j}\ee^{\ii\pi l_1(\mu_j+\frac{\eta}{2})+\ii\xi}\ell{1}^m(\mu_j)}{\ell{1}^m(\eta)}\left[\frac{\theta_1(\mu_j+\frac{\eta}{2})}{\theta_1(\eta)}\right]^L=Q_2(\mu_j)Q_2(\mu_j+\eta),\\
&c\frac{\ee^{\ii\pi \nu_j}\ee^{-\ii\pi l_1(\nu_j-\frac{\eta}{2})-\ii\xi}\ell{1}^m(\nu_j)}{\ell{1}^m(\eta)}\left[\frac{\theta_1(\nu_j-\frac{\eta}{2})}{\theta_1(\eta)}\right]^L=-Q_1(\nu_j)Q_1(\nu_j-\eta),\\
&\ee^{\ii L\xi}\left[\frac{Q_1(-\frac{\eta}{2})}{Q_2(\frac{\eta}{2})}\right]^L=\pm 1.
\end{aligned}
\end{align}
The energy of the system now reads
\begin{align}
E(\mathbf{\mu},\mathbf{\nu})=2\sum_{j=1}^M[f(\nu_j-\tfrac{\eta}{2})-f(\mu_j+\tfrac{\eta}{2})]+Lf(\eta)+2\ii\pi l_1\frac{\ell{1}(\eta)}{\ell{1}'(0)}.
\end{align}


\paragraph{Selection rule}
For a non-zero overlap between the boundary state and a Bethe state, the condition $\Lambda(u)=\Lambda(-u)$ or $\Lambda(u)=-\Lambda(-u)$ must hold. This condition gives a specific selection rule on the Bethe roots.
Unfortunately, no simple pattern is evident for the Bethe roots in this case. Nevertheless, we can obtain a selection rule by analyzing $\Lambda(u)$ at $u=\pm\eta/2$ (cf. Eq.~\eqref{constraint:z:plus}).
The function $\Lambda(u)$ at the points $u=\pm \eta/2$ can be calculated 
\begin{align}
\Lambda(\tfrac{\eta}{2})=\ee^{\ii\xi}\frac{Q_1(-\frac{\eta}{2})}{Q_2(\frac{\eta}{2})},\quad 
\Lambda(-\tfrac{\eta}{2})=\ee^{-\ii\xi}(-1)^{L+1}\frac{Q_2(\frac{\eta}{2})}{Q_1(-\frac{\eta}{2})}.
\end{align}
Imposing $\Lambda(u)=\Lambda(-u)$ yields the following necessary condition for a non‑zero overlap between the Bethe state and the boundary state
\begin{align}
\ee^{\ii\xi}\frac{Q_1(-\frac{\eta}{2})}{Q_2(\frac{\eta}{2})}=\pm \ii^{L+1}, \label{constraint:xyz:1}
\end{align}
Once we require $\Lambda(u)=-\Lambda(-u)$, the following necessary constraint must be satisfied
\begin{align}
\ee^{\ii\xi}\frac{Q_1(-\frac{\eta}{2})}{Q_2(\frac{\eta}{2})}=\pm \ii^{L}, \label{constraint:xyz:2}
\end{align}
Here, we only consider the values of $\Lambda(u)$ at $u=\pm \eta/2$. Considering the $k$-th derivative of $\Lambda(u)$ should yield further constraints on the Bethe roots.



\section{Conclusion and discussion}
\label{sec6}
Previous studies on integrable boundary states have focused on the periodic spin chain with even number of sites. Our work presents a new family of generalized integrable boundary states in $\mathrm{XXZ}$ and $\mathrm{XYZ}$ chains under both periodic and twisted boundary conditions (summarized in Tables \ref{tab1} and \ref{tab2}). Our generalized integrable boundary states differ from conventional ones in two key aspects: they exist for odd‑site systems, and their defining condition generalizes from $t(u) |\Psi\rangle= t(-u) |\Psi\rangle$ to $t(u) |\Psi\rangle=\pm t(-u) |\Psi\rangle$.  
Based on the variant of $\mathrm{KT}$-relation, all the boundary states can be proved in a simple way.

The selection rules given by the boundary states is also studied. In most cases, the Bethe roots in the $T$-$Q$ relation must form pairs to guarantee a non-zero overlap between the Bethe state and the integrable boundary state. However, this paired pattern fails to appear in the following settings: (i) $|\Psi_{+,e}\rangle$ for the XXZ chain with twisted boundary; (ii) $|\Psi_{\pm,e}\rangle$ and $|\Psi_{+,o}\rangle$ for the XYZ chain with twisted boundary. In these specific cases, the selection rules for the Bethe roots are obtained by analyzing the behavior of $\Lambda(u)$ at the points $u=\pm\eta/2$ (see Eqs. (\ref{constraint:z:plus}), (\ref{selection:rule:z}), (\ref{selection:rule:x1}), (\ref{constraint:xyz:1}) and (\ref{constraint:xyz:2})).

The selection rule need to be discussed further. Let $\mathcal{F}^\pm$ be the set of Bethe states $|\phi\rangle$ for which $\langle\Psi_{\pm,s}|\phi\rangle\neq 0,\,s=e\,\mbox{or}\, o$. The dimension of $\mathcal{F}^\pm$ is less than the dimension of the Hilbert space, while its exact value depends on the explicit form of the boundary state. In this paper, what we investigate is another enlarged set $\mathcal{F}_u^\pm$ ($\mathcal{F}_{u}^\pm\supseteq\mathcal{F}^\pm$), defined as the collection of Bethe states $|\phi\rangle$ satisfying $t(u)|\phi\rangle=\pm t(-u)|\phi\rangle$. We numerically compute the dimension of $\mathcal{F}_u^\pm$ for the XXZ and XYZ chains under various integrable boundary conditions and present the results in Table \ref{tab3}. We conclude that 
\begin{align}
\rm{dim} (\mathcal{F}_{u}^-)=\sqrt{\rm{dim} (\mathcal{H})},\,\,\mbox{or}\,\,2\sqrt{\rm{dim} (\mathcal{H})},
\end{align}
except when $G=\mathbb{I}$ and $L$ is even, for which we do not obtain a closed formula. This result implies that the boundary state \textit{significantly} restricts the set of Bethe states that are relevant for the dynamical issue.  In some cases, a necessary and sufficient condition for $t(u)|\phi\rangle=\pm t(-u)|\phi\rangle$ cannot be expressed in an elegant closed form. Nevertheless, useful constraints can be obtained by analyzing $\Lambda(u)$ at the points $u=\pm\eta/2$. Let us introduce another set  $\mathcal{F}_\eta^\pm$ ($\mathcal{F}_{\eta}^\pm\supseteq\mathcal{F}_{u}^\pm$), defined as the collection of Bethe states $|\phi\rangle$ satisfying $t(\frac{\eta}{2})|\phi\rangle=\pm t(-\frac{\eta}{2})|\phi\rangle$. Note that the operator $t(\frac{\eta}{2})=G_1R_{1,N}\dots R_{1,3}P_{1,2}$ acts as a shift operator. Since $t(\frac{\eta}{2})t(-\frac{\eta}{2})\!=\!(-1)^k\times \mathbb{I}$, the condition $t(\frac{\eta}{2})|\phi\rangle=\pm t(-\frac{\eta}{2})|\phi\rangle$ forces the total momentum to take specific quantized values. The numerical results for the dimension of $\mathcal F^\pm_\eta$ are summarized in Table \ref{tab4}.  For small system sizes ($L\leq 6$), the sets  $\mathcal{F}_{\eta}^\pm$ and $\mathcal{F}_{u}^\pm$ nearly coincide, except when $G=\mathbb{I}$ and $L$ is a multiple of four; in that case $\mathcal{F}_{\eta}^-\neq\emptyset,\,\,\mathcal{F}_{u}^-=\emptyset$ (see the red parts in Table \ref{tab4}).  We conclude (except the case where $G=\mathbb{I}$ and $L=4k,\,k\in\mathbb{N}^+$)
\begin{align}
\rm{dim} (\mathcal{F}_{\eta}^-)\sim L^{-1}{\rm{dim} (\mathcal{H})},\,\,\mbox{or}\,\,2L^{-1}{\rm{dim} (\mathcal{H})}.
\end{align}
This result demonstrates that even the zero‑order derivative of $\Lambda(u)$ (i.e., its value) at specific point $u=\frac{\eta}{2}$ is sufficient to enforce a strict constraint on the Bethe state.

\begin{table}[htbp]
\centering
\begin{tabular}{|c|c|c|c|c|c|c|c|c|}
\hline
\diagbox{G} {$N_+$} {$L$}& 3& 4 & 5 & 6 & 7 & 8 & 9 & 10\\[4pt]
\hline
$\mathbb{I}$ &0  &10 & 0& 24& 0& 54& 0& 120\\[4pt]
\hline 
${\sigma^x}$ &4& 4& 8& 8& 16& 16& 32& 32\\[4pt]
\hline 
${\sigma^y}$ &4& 4& 8& 8& 16& 16& 32& 32\\[4pt]
\hline 
${\sigma^z}$ &4& 4& 8& 8& 16& 16& 32& 32\\[4pt]
\hline 
\end{tabular}

~~~

~~~~

\begin{tabular}{|c|c|c|c|c|c|c|c|c|}
\hline
\diagbox{G} {$N_-$} {$L$}& 3& 4 & 5 & 6 & 7 & 8 & 9 & 10\\[4pt]
\hline
$\mathbb{I}$ &4& 0& 8& 0& 16& 0& 32& 0\\[4pt]
\hline 
${\sigma^x}$ &0& 4& 0& 8& 0& 16& 0& 32\\[4pt]
\hline 
${\sigma^y}$ &0& 4& 0& 8& 0& 16& 0& 32\\[4pt]
\hline 
${\sigma^z}$ &0& 4& 0& 8& 0& 16& 0& 32\\[4pt]
\hline 
\end{tabular}
\caption{Zero eigenvalue degeneracy $N_\pm$ of the operator $t(u)\mp t(-u)$ versus system length $L$ and twist matrix $G$. We obtain the results numerically by setting the spectral parameter $u$ and the crossing parameter $\eta$ as random parameters. All results in this table hold for both the XXZ and XYZ chains.}
\label{tab3}
\end{table}
\begin{table}[htbp]
\centering
\begin{tabular}{|c|c|c|c|c|c|c|c|c|}
\hline
\diagbox{G} {~~\\$\widetilde N_+$} {$L$}& 3& 4 & 5 & 6 & 7 & 8 & 9 & 10\\[4pt]
\hline
$\mathbb{I}$ &0  &10 & 0& 24& 0& 70& 0& 208\\[4pt]
\hline 
${\sigma^x}$ &4& 4& 8& 12& 20& 32& 60& 104\\[4pt]
\hline 
${\sigma^y}$ &4& 4& 8& 12& 20& 32& 60& 104\\[4pt]
\hline 
${\sigma^z}$ &4& 4& 8& 12& 20& 32& 60& 104\\[4pt]
\hline 
\end{tabular}

\vspace{0.5cm}

\begin{tabular}{|c|c|c|c|c|c|c|c|c|}
\hline
\diagbox{G} {~~\\$\widetilde N_-$} {$L$}& 3& 4 & 5 & 6 & 7 & 8 & 9 & 10\\[4pt]
\hline
$\mathbb{I}$ &4& \color{red}{6}& 8& 0& 20& \color{red}{66}& 60& 0\\[4pt]
\hline 
${\sigma^x}$ &0& 4& 0& 12& 0& 32& 0& 104\\[4pt]
\hline 
${\sigma^y}$ &0& 4& 0& 12& 0& 32& 0& 104\\[4pt]
\hline 
${\sigma^z}$ &0& 4& 0& 12& 0& 32& 0& 104\\[4pt]
\hline 
\end{tabular}
\caption{Zero eigenvalue degeneracy $\widetilde N_\pm$ of the operator $t(\frac{\eta}{2})\mp t(-\frac{\eta}{2})$ versus system length $L$ and twist matrix $G$. We obtain the results numerically by setting the crossing parameter $\eta$ to a random value. All results in this table hold for both the XXZ and XYZ chains.}
\label{tab4}
\end{table}

A natural follow-up question is the explicit overlap between these integrable boundary states and Bethe states. An open question is whether such overlap takes the same universal term, namely the ratio of Gaudin-like determinant \cite{Pozsgay_2014, Brockmann_2014}, as in rational case. Some numerical evidence has shown that such Gaudin-like determinant will be present in odd site periodic $\mathrm{XXZ}$ chain \cite{Pozsgay_2018}. Nevertheless, this question is challenging because the explicit form of the Bethe states is not fully accessible in certain cases. In the XXZ chain with non-diagonal twist, Bethe states are known yet their norms and scalar products are undetermined \cite{Zhang_2015}. Likewise, for the periodic XYZ chain with odd $L$ or with twisted boundaries, an explicit construction of Bethe states is still lacking.

In this work, we focus on constructing integrable boundary states via the KT relation, restricting these states to the two‑site factorized form. However, it has been discussed that some other integrable states beyond these two-site states can still exist, such as the crosscap state \cite{Caetano_2022} and Dicke states \cite{Nepomechie_2024}
\footnote{In certain cases, Dicke states can be obtained directly from our two-point boundary states. For example, in the periodic case, the product state \( \binom{1}{a}^{\otimes L} \) with arbitrary parameter \( a \) constitutes an integrable boundary state. Expanding this state in powers of \( a \) yields a set of Dicke states, which are also integrable boundary states.}.  All the boundary states constructed in this paper can annihilate certain parity‑symmetric charges. This restriction, though simple, is not the only possibility. There may exist other states that can select certain Bethe states, but one can not find such an obvious restriction on the transfer matrix. The existence of other (generalized) integrable states could be a way of proceeding for future research. 

The integrable boundary states constructed in this paper have a remarkable simple factorized structure, which makes it more likely to be prepared in cold atom experiments \cite{Guan_2013,Jepsen:2021,Jepsen:2021_2}. These states are idea examples to study non-thermalizing dynamics, in quantum many-body systems, both analytically and experimentally.  
Another interesting question is to see whether these integrable boundary states retain a simple structure in the thermodynamic limit.

\section*{Acknowledgments}
Xin Qian gratefully acknowledges Prof. Yunfeng Jiang for many helpful discussions and revisions of the paper. He also wants to acknowledge the hospitality of the Niels Bohr International Academy and Shi-Tung Yau Center of Southeast University where a part of this work was carried out. The work of Xin Qian was supported by CSC through grant number 20220794001. Xin Zhang acknowledges financial support from the National Natural Science Foundation of China (No. 12575007). Both authors thank Xi-Wen Guan, Balázs Pozsgay and Yuan Miao for their valuable comments.

\appendix
\section{Theta function convention}
\label{Appendix_A}
We use the Jacobi $\theta$-function with definition as
\begin{equation}
    \begin{split}
\theta_1(u\vert\tau)&=-\ii\sum_{k\in \mathbb{Z}}(-1)^kq^{(k+1/2)^2}\ee^{\ii\pi(2k+1)u},\\
\theta_2(u\vert\tau)&=\sum_{k\in \mathbb{Z}}q^{(k+1/2)^2}\ee^{\ii\pi(2k+1)u},\\
\theta_3(u\vert\tau)&=\sum_{k\in \mathbb{Z}}q^{k^2}\ee^{2\ii\pi ku},\\
\theta_4(u\vert\tau)&=\sum_{k\in \mathbb{Z}}(-1)^kq^{k^2}\ee^{2\ii\pi ku},\\
\end{split}
\end{equation}
where $\tau\in\mathbb{C},\mathrm{Im}(\tau)>0$ and $q=\ee^{\ii\pi\tau}$. And we will use $\vartheta_i(u):=\theta(u\lvert 2\tau)$ for short hand notation. 
\section{Integrable final state from Boundary Yang-Baxter equation}
\label{Appendix_B}
The boundary Yang-Baxter equation Eq. (\ref{BYBeq}) can be reformulated into the KYB equation, which is proved in Appendix \ref{Appendix_C},
\begin{equation}
    \vec{K}_{1,2}(u)\vec{K}_{3,4}(v+\tfrac{\eta}{2})\mathcal{L}_{1,4}(u+v)\mathcal{L}_{1,3}(u-v)=\vec{K}_{1,2}(u)\vec{K}_{3,4}(v+\tfrac{\eta}{2})\mathcal{L}_{2,3}(u+v)\mathcal{L}_{2,4}(u-v),
    \label{KYBEQ1}
\end{equation}
where we use the convention in \cite{Gombor_2021} that $\vec{K}$ is a re-arrangement of $K(u)$ into a co-vector as
\begin{equation}
    \vec{K}(u)=\sum_{i,j}\tilde{k}_{ij}(u)e_i^{*}\otimes e_j^{*},\quad \tilde{k}_{ij}=\left[\sigma^{y}K(-u)\right]_{ij},
\end{equation}
with the column vector $e_1^{*}=\begin{pmatrix}
    1\\0
\end{pmatrix}$, $e_2^{*}=\begin{pmatrix}
    0\\1
\end{pmatrix}$.
From $\mathrm{KYB}$ equation, we can now define integrable finial states as
\begin{equation}
    \langle \Psi\lvert=\vec{K}(\tfrac{\eta}{2})^{\otimes \frac{L}{2}},
\end{equation}
one can easily get the so called $\mathrm{KT}$-relation,
\begin{equation}
    \langle\Psi\lvert\vec{K}_{12}(u)T_{1}(u)=\langle\Psi\lvert\vec{K}_{12}(u)T^{\pi}_{2}(u),
    \label{KT1}
\end{equation}
where $T_0^{\pi}=\mathcal{L}_{0,1}(u)\dots \mathcal{L}_{0,L}(u)$ is the space reflected version of the monodromy matrix, $T(u)$ is the monodromy matrix. Eq. (\ref{KT1}) can be rewritten as
\begin{equation}
    \langle\Psi\lvert \sigma^yK^{\mathrm{t}}(-u)\sigma^yT(u)=\langle\Psi\lvert T(-u)\sigma^yK^{\mathrm{t}}(-u)\sigma^y.
    \label{KT2}
\end{equation}
By taking a trace over the auxiliary space, one can easily prove the integrability condition $\langle\Psi\lvert\left(t(u)-t(-u)\right)=0$.
\section{The equivalence between KYB equation and BYB equation}
\label{Appendix_C}
Given the operator $V$ acting on $h_1$, we use the following convention for its matrix elements \cite{Piroli_2017}
\begin{equation}
    \langle i\lvert V\lvert j\rangle=V_{ij}=V^i_j,
\end{equation}
so that
\begin{equation}
    \langle i\lvert V=\langle j\lvert V^i_{j}.
\end{equation}
Analogously, given the operator $R_{12}(u)$ acting on $h_1\otimes h_2$, we use
\begin{equation}
    {}_1{\langle} i\lvert{}_2\langle j\lvert R_{12}(u)\lvert l\rangle_1\lvert k\rangle_2=R_{lk}^{ij}(u),
\end{equation}
so that
\begin{equation}
    \langle i,j\lvert R_{12}(u)=R^{ij}_{lk}\langle l,k\lvert.
\end{equation}
Using these conventions, the matrix elements for the partial transposed are easily written 
\begin{equation}
    {}_1{\langle} i\lvert{}_2\langle j\lvert R^{t_1}_{12}(u)\lvert l\rangle_1\lvert k\rangle_2=R_{ik}^{lj}(u).
\end{equation}
The $\mathrm{KYB}$ equation reads 
\begin{equation}
    \vec{K}_{12}(u)\vec{K}_{34}(v+\tfrac{\eta}{2})\mathcal{L}_{14}(u+v)\mathcal{L}_{13}(u-v)=\vec{K}_{12}(u)\vec{K}_{34}(v+\tfrac{\eta}{2})\mathcal{L}_{23}(u+v)\mathcal{L}_{24}(u-v).
    \label{B6}
\end{equation}
The LHS of Eq. (\ref{B6}) is
\begin{equation}
    \begin{split}
    \mathrm{LHS}&=\langle i_1\lvert\otimes\langle i_2\lvert\otimes\langle i_3\lvert\otimes\langle i_4\lvert \tilde{k}_{i_1,i_2}(u)\tilde{k}_{i_3,i_4}(v+\tfrac{\eta}{2})\mathcal{L}_{14}(u+v)\mathcal{L}_{13}(u-v)\\
&=\langle k_1\lvert\otimes\langle i_2\lvert\otimes\langle j_3\lvert\otimes\langle j_4\lvert \tilde{k}_{i_1,i_2}(u)\tilde{k}_{i_3,i_4}(v+\tfrac{\eta}{2})\mathcal{L}_{j_1j_4}^{i_1i_4}(u+v)\mathcal{L}_{k_1j_3}^{j_1i_3}(u-v)\\
&=\langle j_1\lvert\otimes\langle j_2\lvert\otimes\langle j_3\lvert\otimes\langle j_4\lvert \tilde{k}_{i_1,j_2}(u)\tilde{k}_{i_3,i_4}(v+\tfrac{\eta}{2})\mathcal{L}_{i_2j_4}^{i_1i_4}(u+v)\mathcal{L}_{j_1j_3}^{i_2i_3}(u-v),
\end{split}
\end{equation}
where the RHS of Eq. (\ref{B6}) is
\begin{equation}
    \begin{split}
\mathrm{RHS}=\langle j_1\lvert\otimes\langle j_2\lvert\otimes\langle j_3\lvert\otimes\langle j_4\lvert \tilde{k}_{j_1,i_2}(u)\tilde{k}_{i_3,i_4}(v+\tfrac{\eta}{2})\mathcal{L}_{i_1j_3}^{i_2i_3}(u+v)\mathcal{L}_{j_2j_4}^{i_1i_4}(u-v).
\label{C8}
\end{split}
\end{equation}
Then we will get
\begin{align}
\tilde{k}_{i_1,j_2}(u)\tilde{k}_{i_3,i_4}(v+\tfrac{\eta}{2})\mathcal{L}_{i_2j_4}^{i_1i_4}(u+v)\mathcal{L}_{j_1j_3}^{i_2i_3}(u-v)\no\\
=\tilde{k}_{j_1,i_2}(u)\tilde{k}_{i_3,i_4}(v+\tfrac{\eta}{2})\mathcal{L}_{i_1j_3}^{i_2i_3}(u+v)\mathcal{L}_{j_2j_4}^{i_1i_4}(u-v).
\end{align}
Rewrite $\langle i\lvert\otimes\langle j\lvert\tilde{k}_{ij}=\langle i\lvert\otimes\langle j\lvert\tilde{k}_{i}^{j}$, we went back to the matrix convention, then (\ref{C8}) yields
\begin{equation}
    \begin{split}
&\tilde{k}_{i_1}^{j_2}(u)\tilde{k}_{i_3}^{i_4}(v+\tfrac{\eta}{2})\mathcal{L}_{i_2j_4}^{i_1i_4}(u+v)\mathcal{L}_{j_1j_3}^{i_2i_3}(u-v)=\tilde{k}_{j_1}^{i_2}(u)\tilde{k}_{i_3}^{i_4}(v+\tfrac{\eta}{2})\mathcal{L}_{i_1j_3}^{i_2i_3}(u+v)\mathcal{L}_{j_2j_4}^{i_1i_4}(u-v)\\
\Leftrightarrow&\tilde{k}_{i_1}^{j_2}(u)\tilde{k}_{i_3}^{i_4}(v+\tfrac{\eta}{2})\mathcal{L}_{i_2i_4}^{i_1j_4}(u+v)\mathcal{L}_{j_1j_3}^{i_2i_3}(u-v)=\tilde{k}_{j_1}^{i_2}(u)\tilde{k}_{i_3}^{i_4}(v+\tfrac{\eta}{2})\mathcal{L}_{i_2j_3}^{i_1i_3}(u+v)\mathcal{L}_{i_1i_4}^{j_2j_4}(u-v)\\
\Leftrightarrow&{}_{1}\langle j_2\lvert{}_{2}\langle j_4\lvert \left(\sigma^{y}K(-u)\right)_1\mathcal{L}_{12}^{t_2}(u+v)(\sigma^{y}K(-v-\tfrac{\eta}{2}))_{2}\mathcal{L}_{12}(u-v)\lvert j_1\rangle_1\lvert j_3\rangle_2\\&={}_{1}\langle j_2\lvert{}_{2}\langle j_4\lvert \mathcal{L}_{12}^{t_1t_2}(u-v)(\sigma^{y}K(-v-\tfrac{\eta}{2}))_{2}\mathcal{L}^{t_1}_{12}(u+v)\left(\sigma^{y}K(-u)\right)_1\lvert j_1\rangle_1\lvert j_3\rangle_2,
\end{split}
\end{equation}
where in the second line we use a transpose. Then use the crossing relation for the Lax operator, one can easily get
\begin{align}
\sigma_{2}^{y}\left(\sigma^{y}K(v)\right)_2\mathcal{L}_{12}(u+v)\sigma_{1}^{y}\left(\sigma^{y}K(u-\tfrac{\eta}{2})\right)_1
\mathcal{L}_{12}(u-v)\no\\
=\mathcal{L}_{12}(u-v)\sigma_{1}^{y}\left(\sigma^{y}K(u-\tfrac{\eta}{2})\right)_1\mathcal{L}_{12}(u+v)\sigma_2\left(\sigma^{y}K(v)\right)_2,
\end{align}
use the relation $\mathcal{L}_{0,j}(u)=R_{0,j}(u-\eta/2)$, one can get the BYB Equation Eq. (\ref{BYBeq}). Here the proof works for both $\mathrm{XXZ}$ spin chain and $\mathrm{XYZ}$ spin chain.

\section{KT-relation in components of XXZ and XYZ chain}
\label{Appendix_D}
The $\mathrm{KT}$-relation Eq. (\ref{KT3}) is useful for the proof of integrable condition in the presence of twist. Expanding in components of Eq. (\ref{KT3}) we have 
\begin{align}\label{TT:Relation}
\begin{aligned}
\left[k_{22}(-u) (A(-u)-A(u))-k_{21}(-u) B(-u)+k_{12}(-u) C(u)\right]|\Psi_0\rangle=0,\\
\left[k_{11}(-u) B(-u)-k_{22}(-u) B(u)+k_{12}(-u) (D(u)-A(-u))\right]|\Psi_0\rangle=0,\\
\left[k_{21}(-u) (A(u)-D(-u))-k_{11}(-u) C(u)+k_{22}(-u) C(-u)\right]|\Psi_0\rangle=0,\\
\left[k_{21}(-u) B(u)-k_{12}(-u) C(-u)+k_{11}(-u) (D(-u)-D(u))\right]|\Psi_0\rangle=0,
\end{aligned}
\end{align}
which means that we can solve $A(-u),B(-u),C(-u),D(-u)$ reversely
\begin{equation}
    \begin{split}
A(-u)\lvert\Psi_0\rangle=&-\frac{1}{\mathfrak{K}(u)}\big(-A(u)k_{11}(-u)k_{22}(-u)-B(u)k_{21}(-u)k_{22}(-u)\\ & +C(u)k_{11}(-u)k_{12}(-u)+D(u)k_{12}(-u)k_{21}(-u)\big)\lvert\Psi_0\rangle,\\
B(-u)\lvert\Psi_0\rangle=&\frac{1}{\mathfrak{K}(u)}\big(A(u)k_{22}(-u)k_{12}(-u)+B(u)k_{22}(-u)k_{22}(-u)\\
&-C(u)k_{12}(-u)k_{12}(-u)-D(u)k_{12}(-u)k_{22}(-u)\big)\lvert\Psi_0\rangle,\\
C(-u)\lvert\Psi_0\rangle=&\frac{1}{\mathfrak{K}(u)}\big(-A(u)k_{21}(-u)k_{11}(-u)-B(u)k_{21}(-u)k_{21}(-u)\\ &+C(u)k_{11}(-u)k_{11}(-u)+D(u)k_{11}(-u)k_{21}(-u)\big),\\
D(-u)\lvert\Psi_0\rangle=&\frac{1}{\mathfrak{K}(u)}\big(-A(u)k_{21}(-u)k_{12}(-u)-B(u)k_{21}(-u)k_{22}(-u)\\&+C(u)k_{11}(-u)k_{12}(-u)+D(u)k_{11}(-u)k_{22}(-u)\big)\lvert\Psi_0\rangle,
\label{tmucomp}
\end{split}
\end{equation}
with $\mathfrak{K}(u)=k_{11}(-u)k_{22}(-u)-k_{12}(-u)k_{21}(-u)$, where $k_{11}(u),\ k_{12}(u),\ k_{21}(u),\ k_{22}(u)$ is defined in Eq. (\ref{K:XXZ}) with the following relation
\begin{equation}
    \begin{split}
    k_{11}(u)=\frac{\alpha_4}{k_4(u)}+\frac{\alpha_3}{k_3(u)},\quad k_{12}(u)=\frac{\alpha_1}{k_1(u)}-\frac{\ii\alpha_2}{k_2(u)},\\
    k_{21}(u)=\frac{\alpha_1}{k_1(u)}+\frac{\ii\alpha_2}{k_2(u)},\quad k_{22}(u)=\frac{\alpha_4}{k_4(u)}-\frac{\alpha_3}{k_3(u)}.
    \end{split}
\end{equation}





\bibliography{SciPost_Example_BiBTeX_File.bib}

@article{PhysRevLett.19.1312,
  title = {{Some Exact Results for the Many-Body Problem in one Dimension with Repulsive Delta-Function Interaction}},
  author = {Yang, C. N.},
  journal = {Phys. Rev. Lett.},
  volume = {19},
  issue = {23},
  pages = {1312--1315},
  numpages = {0},
  year = {1967},
  month = {Dec},
  publisher = {American Physical Society},
  doi = {10.1103/PhysRevLett.19.1312},
  url = {https://link.aps.org/doi/10.1103/PhysRevLett.19.1312}
}

@book{Korepin_Bogoliubov_Izergin_1993, place={Cambridge}, series={Cambridge Monographs on Mathematical Physics}, title={{Quantum Inverse Scattering Method and Correlation Functions}}, publisher={Cambridge University Press}, author={Korepin, V. E. and Bogoliubov, N. M. and Izergin, A. G.}, year={1993}, collection={Cambridge Monographs on Mathematical Physics}}

@misc{faddeev1996algebraicbetheansatzworks,
      title={{How Algebraic Bethe Ansatz works for integrable model}}, 
      author={L. D. Faddeev},
      year={1996},
      eprint={hep-th/9605187},
      archivePrefix={arXiv},
      primaryClass={hep-th},
      url={https://arxiv.org/abs/hep-th/9605187}, 
}

@book{baxter1982,
  author    = {Baxter, R. J.},
  title     = {{Exactly Solved Models in Statistical Mechanics}},
  year      = {1982},
  publisher = {Academic Press},
  address   = {London}
}

@book{Wang2015,
  title     = {{Off-Diagonal Bethe Ansatz for Exactly Solvable Models}},
  author    = {Wang, Yupeng and Yang, Wen-Li and Cao, Junpeng and Shi, Kangjie},
  year      = {2015},
  publisher = {Springer},
  address   = {Berlin, Heidelberg},
  isbn      = {978-3-662-46756-5},
  doi       = {10.1007/978-3-662-46756-5}
}

@article{Cao_2013,
   title={{Off-Diagonal Bethe Ansatz and Exact Solution of a Topological Spin Ring}},
   volume={111},
   ISSN={1079-7114},
   url={http://dx.doi.org/10.1103/PhysRevLett.111.137201},
   DOI={10.1103/physrevlett.111.137201},
   number={13},
   journal={Physical Review Letters},
   publisher={American Physical Society (APS)},
   author={Cao, Junpeng and Yang, Wen-Li and Shi, Kangjie and Wang, Yupeng},
   year={2013},
   month=sep }

@article{Cao_2014,
   title={{Spin-1/2 XYZ model revisit: General solutions via off-diagonal Bethe ansatz}},
   volume={886},
   ISSN={0550-3213},
   url={http://dx.doi.org/10.1016/j.nuclphysb.2014.06.026},
   DOI={10.1016/j.nuclphysb.2014.06.026},
   journal={Nuclear Physics B},
   publisher={Elsevier BV},
   author={Cao, Junpeng and Cui, Shuai and Yang, Wen-Li and Shi, Kangjie and Wang, Yupeng},
   year={2014},
   month=sep, pages={185–201} }

@article{Cao_2015,
   title={{On the complete-spectrum characterization of quantum integrable spin chains via inhomogeneous $T$–$Q$ relation}},
   volume={48},
   ISSN={1751-8121},
   url={http://dx.doi.org/10.1088/1751-8113/48/44/444001},
   DOI={10.1088/1751-8113/48/44/444001},
   number={44},
   journal={Journal of Physics A: Mathematical and Theoretical},
   publisher={IOP Publishing},
   author={Cao, Junpeng and Yang, Wen-Li and Shi, Kangjie and Wang, Yupeng},
   year={2015},
   month=oct, pages={444001} }

@article{Bloch_2008,
   title={{Many-body physics with ultracold gases}},
   volume={80},
   ISSN={1539-0756},
   url={http://dx.doi.org/10.1103/RevModPhys.80.885},
   DOI={10.1103/revmodphys.80.885},
   number={3},
   journal={Reviews of Modern Physics},
   publisher={American Physical Society (APS)},
   author={Bloch, Immanuel and Dalibard, Jean and Zwerger, Wilhelm},
   year={2008},
   month=jul, pages={885–964} }

@article{Polkovnikov_2011,
   title={{Colloquium: Nonequilibrium dynamics of closed interacting quantum systems}},
   volume={83},
   ISSN={1539-0756},
   url={http://dx.doi.org/10.1103/RevModPhys.83.863},
   DOI={10.1103/revmodphys.83.863},
   number={3},
   journal={Reviews of Modern Physics},
   publisher={American Physical Society (APS)},
   author={Polkovnikov, Anatoli and Sengupta, Krishnendu and Silva, Alessandro and Vengalattore, Mukund},
   year={2011},
   month=aug, pages={863–883} }

@article{Guan_2013,
   title={{Fermi gases in one dimension: From Bethe ansatz to experiments}},
   volume={85},
   ISSN={1539-0756},
   url={http://dx.doi.org/10.1103/RevModPhys.85.1633},
   DOI={10.1103/revmodphys.85.1633},
   number={4},
   journal={Reviews of Modern Physics},
   publisher={American Physical Society (APS)},
   author={Guan, Xi-Wen and Batchelor, Murray T. and Lee, Chaohong},
   year={2013},
   month=nov, pages={1633–1691} }

@article{Rigol_2007,
   title={{Relaxation in a Completely Integrable Many-Body Quantum System: An AbInitio Study of the Dynamics of the Highly Excited States of 1D Lattice Hard-Core Bosons}},
   volume={98},
   ISSN={1079-7114},
   url={http://dx.doi.org/10.1103/PhysRevLett.98.050405},
   DOI={10.1103/physrevlett.98.050405},
   number={5},
   journal={Physical Review Letters},
   publisher={American Physical Society (APS)},
   author={Rigol, Marcos and Dunjko, Vanja and Yurovsky, Vladimir and Olshanii, Maxim},
   year={2007},
   month=feb }

@article{PhysRevLett.97.156403,
  title = {{Effect of Suddenly Turning on Interactions in the Luttinger Model}},
  author = {Cazalilla, M. A.},
  journal = {Phys. Rev. Lett.},
  volume = {97},
  issue = {15},
  pages = {156403},
  numpages = {4},
  year = {2006},
  month = {Oct},
  publisher = {American Physical Society},
  doi = {10.1103/PhysRevLett.97.156403},
  url = {https://link.aps.org/doi/10.1103/PhysRevLett.97.156403}
}

@article{Cramer_2008,
   title={{Exact Relaxation in a Class of Nonequilibrium Quantum Lattice Systems}},
   volume={100},
   ISSN={1079-7114},
   url={http://dx.doi.org/10.1103/PhysRevLett.100.030602},
   DOI={10.1103/physrevlett.100.030602},
   number={3},
   journal={Physical Review Letters},
   publisher={American Physical Society (APS)},
   author={Cramer, M. and Dawson, C. M. and Eisert, J. and Osborne, T. J.},
   year={2008},
   month=jan }

@article{Calabrese_2011,
   title={{Quantum Quench in the Transverse-Field Ising Chain}},
   volume={106},
   ISSN={1079-7114},
   url={http://dx.doi.org/10.1103/PhysRevLett.106.227203},
   DOI={10.1103/physrevlett.106.227203},
   number={22},
   journal={Physical Review Letters},
   publisher={American Physical Society (APS)},
   author={Calabrese, Pasquale and Essler, Fabian H. L. and Fagotti, Maurizio},
   year={2011},
   month=jun }

@article{Caux_2012,
   title={{Constructing the Generalized Gibbs Ensemble after a Quantum Quench}},
   volume={109},
   ISSN={1079-7114},
   url={http://dx.doi.org/10.1103/PhysRevLett.109.175301},
   DOI={10.1103/physrevlett.109.175301},
   number={17},
   journal={Physical Review Letters},
   publisher={American Physical Society (APS)},
   author={Caux, Jean-Sébastien and Konik, Robert M.},
   year={2012},
   month=oct }

@article{Pozsgay_2013,
   title={{The generalized Gibbs ensemble for Heisenberg spin chains}},
   volume={2013},
   ISSN={1742-5468},
   url={http://dx.doi.org/10.1088/1742-5468/2013/07/P07003},
   DOI={10.1088/1742-5468/2013/07/p07003},
   number={07},
   journal={Journal of Statistical Mechanics: Theory and Experiment},
   publisher={IOP Publishing},
   author={Pozsgay, Balázs},
   year={2013},
   month=jul, pages={P07003} }

@article{Wouters_2014,
   title={{Quenching the Anisotropic Heisenberg Chain: Exact Solution and Generalized Gibbs Ensemble Predictions}},
   volume={113},
   ISSN={1079-7114},
   url={http://dx.doi.org/10.1103/PhysRevLett.113.117202},
   DOI={10.1103/physrevlett.113.117202},
   number={11},
   journal={Physical Review Letters},
   publisher={American Physical Society (APS)},
   author={Wouters, B. and De Nardis, J. and Brockmann, M. and Fioretto, D. and Rigol, M. and Caux, J.-S.},
   year={2014},
   month=sep }

@article{Ilievski_2015,
   title={{Complete Generalized Gibbs Ensembles in an Interacting Theory}},
   volume={115},
   ISSN={1079-7114},
   url={http://dx.doi.org/10.1103/PhysRevLett.115.157201},
   DOI={10.1103/physrevlett.115.157201},
   number={15},
   journal={Physical Review Letters},
   publisher={American Physical Society (APS)},
   author={Ilievski, E. and De Nardis, J. and Wouters, B. and Caux, J.-S. and Essler, F. H. L. and Prosen, T.},
   year={2015},
   month=oct }

@article{David2006quantum,
  title={{A quantum Newton's cradle}},
  author={Kinoshita, Toshiya and Wenger, Trevor and Weiss, David S},
  journal={Nature},
  volume={440},
  number={7087},
  pages={900--903},
  year={2006},
  month={Apr},
  publisher={Nature Publishing Group},
  doi={10.1038/nature04693},
  url={https://doi.org/10.1038/nature04693}
}

@article{Calabrese_2016,
doi = {10.1088/1742-5468/2016/06/064001},
url = {https://doi.org/10.1088/1742-5468/2016/06/064001},
year = {2016},
month = {jun},
publisher = {IOP Publishing and SISSA},
volume = {2016},
number = {6},
pages = {064001},
author = {Calabrese, Pasquale and Essler, Fabian H L and Mussardo, Giuseppe},
title = {{Introduction to ‘Quantum Integrability in Out of Equilibrium Systems’}},
journal = {Journal of Statistical Mechanics: Theory and Experiment}
}

@article{Bastianello_2022,
doi = {10.1088/1742-5468/ac3e6a},
url = {https://doi.org/10.1088/1742-5468/ac3e6a},
year = {2022},
month = {jan},
publisher = {IOP Publishing and SISSA},
volume = {2022},
number = {1},
pages = {014001},
author = {Bastianello, Alvise and Bertini, Bruno and Doyon, Benjamin and Vasseur, Romain},
title = {{Introduction to the Special Issue on Emergent Hydrodynamics in Integrable Many-Body Systems}},
journal = {Journal of Statistical Mechanics: Theory and Experiment}
}

@article{Calabrese_2006,
   title={{Time Dependence of Correlation Functions Following a Quantum Quench}},
   volume={96},
   ISSN={1079-7114},
   url={http://dx.doi.org/10.1103/PhysRevLett.96.136801},
   DOI={10.1103/physrevlett.96.136801},
   number={13},
   journal={Physical Review Letters},
   publisher={American Physical Society (APS)},
   author={Calabrese, Pasquale and Cardy, John},
   year={2006},
   month=apr }

@article{Calabrese_2007,
   title={{Quantum quenches in extended systems}},
   volume={2007},
   ISSN={1742-5468},
   url={http://dx.doi.org/10.1088/1742-5468/2007/06/P06008},
   DOI={10.1088/1742-5468/2007/06/p06008},
   number={06},
   journal={Journal of Statistical Mechanics: Theory and Experiment},
   publisher={IOP Publishing},
   author={Calabrese, Pasquale and Cardy, John},
   year={2007},
   month=jun, pages={P06008–P06008} }

@article{Caux_2013,
   title={{Time Evolution of Local Observables After Quenching to an Integrable Model}},
   volume={110},
   ISSN={1079-7114},
   url={http://dx.doi.org/10.1103/PhysRevLett.110.257203},
   DOI={10.1103/physrevlett.110.257203},
   number={25},
   journal={Physical Review Letters},
   publisher={American Physical Society (APS)},
   author={Caux, Jean-Sébastien and Essler, Fabian H. L.},
   year={2013},
   month=jun }

@article{GHOSHAL_1994,
   title={{Boundary S-Matrix and Boundary State in Two-Dimensional Integrable Quantum Field Theory}},
   volume={09},
   ISSN={1793-656X},
   url={http://dx.doi.org/10.1142/S0217751X94001552},
   DOI={10.1142/s0217751x94001552},
   number={21},
   journal={International Journal of Modern Physics A},
   publisher={World Scientific Pub Co Pte Lt},
   author={Ghoshal, Subir and Zamolodchikov, Alexander},
   year={1994},
   month=aug, pages={3841–3885} }

@article{E_K_Sklyanin_1988,
    doi = {10.1088/0305-4470/21/10/015},
    url = {https://doi.org/10.1088/0305-4470/21/10/015},
    year = {1988},
    month = {may},
    publisher = {},
    volume = {21},
    number = {10},
    pages = {2375},
    author = {E K Sklyanin},
    title = {{Boundary conditions for integrable quantum systems}},
    journal = {Journal of Physics A: Mathematical and General}
}

@article{Piroli_2017,
   title={{What is an integrable quench?}},
   volume={925},
   ISSN={0550-3213},
   url={http://dx.doi.org/10.1016/j.nuclphysb.2017.10.012},
   DOI={10.1016/j.nuclphysb.2017.10.012},
   journal={Nuclear Physics B},
   publisher={Elsevier BV},
   author={Piroli, Lorenzo and Pozsgay, Balázs and Vernier, Eric},
   year={2017},
   month=dec, pages={362–402} }

@article{gombor2024exactoverlapsallintegrable,
  title = {{Exact Overlaps for All Integrable Matrix Product States of Rational Spin Chains}},
  author = {Gombor, Tamas},
  journal = {Phys. Rev. Lett.},
  volume = {135},
  issue = {15},
  pages = {150402},
  numpages = {8},
  year = {2025},
  month = {Oct},
  publisher = {American Physical Society},
  doi = {10.1103/4vy2-8cnk},
  url = {https://link.aps.org/doi/10.1103/4vy2-8cnk}
}

@article{gombor2025derivationsmpsoverlapformulas,
  author       = {Gombor, Tamas},
  title        = {{Derivations for the MPS overlap formulas of rational spin chains}},
  journal      = {Journal of High Energy Physics},
  year         = {2025},
  number       = {35},
  month        = {Oct},
  volume       = {2025},
  doi          = {10.1007/JHEP10(2025)035},
  url          = {https://doi.org/10.1007/JHEP10(2025)035},
  received     = {2025-08-29},
  accepted     = {2025-09-10},
  published    = {2025-10-03},
}

@article{de_Leeuw_2015,
   title={{One-point functions in defect CFT and integrability}},
   volume={2015},
   ISSN={1029-8479},
   url={http://dx.doi.org/10.1007/JHEP08(2015)098},
   DOI={10.1007/jhep08(2015)098},
   number={8},
   journal={Journal of High Energy Physics},
   publisher={Springer Science and Business Media LLC},
   author={de Leeuw, Marius and Kristjansen, Charlotte and Zarembo, Konstantin},
   year={2015},
   month=aug }

@article{Buhl_Mortensen_2016,
   title={{One-point functions in AdS/dCFT from matrix product states}},
   volume={2016},
   ISSN={1029-8479},
   url={http://dx.doi.org/10.1007/JHEP02(2016)052},
   DOI={10.1007/jhep02(2016)052},
   number={2},
   journal={Journal of High Energy Physics},
   publisher={Springer Science and Business Media LLC},
   author={Buhl-Mortensen, Isak and de Leeuw, Marius and Kristjansen, Charlotte and Zarembo, Konstantin},
   year={2016},
   month=feb }

@article{Minahan_2003,
   title={{The Bethe-ansatz for $\mathcal{N}$ = 4 super Yang-Mills}},
   volume={2003},
   ISSN={1029-8479},
   url={http://dx.doi.org/10.1088/1126-6708/2003/03/013},
   DOI={10.1088/1126-6708/2003/03/013},
   number={03},
   journal={Journal of High Energy Physics},
   publisher={Springer Science and Business Media LLC},
   author={Minahan, Joseph A and Zarembo, Konstantin},
   year={2003},
   month=mar, pages={013–013} }

@article{Beisert_2011,
   title={{Review of AdS/CFT Integrability: An Overview}},
   volume={99},
   ISSN={1573-0530},
   url={http://dx.doi.org/10.1007/s11005-011-0529-2},
   DOI={10.1007/s11005-011-0529-2},
   number={1–3},
   journal={Letters in Mathematical Physics},
   publisher={Springer Science and Business Media LLC},
   author={Beisert, Niklas and others},
   year={2011},
   month=oct, pages={3–32} }

@article{de_Leeuw_2016,
   title={{AdS/dCFT one-point functions of the SU(3) sector}},
   volume={763},
   ISSN={0370-2693},
   url={http://dx.doi.org/10.1016/j.physletb.2016.10.044},
   DOI={10.1016/j.physletb.2016.10.044},
   journal={Physics Letters B},
   publisher={Elsevier BV},
   author={de Leeuw, Marius and Kristjansen, Charlotte and Mori, Stefano},
   year={2016},
   month=dec, pages={197–202} }

@article{de_Leeuw_2018,
   title={{Scalar one-point functions and matrix product states of AdS/dCFT}},
   volume={781},
   ISSN={0370-2693},
   url={http://dx.doi.org/10.1016/j.physletb.2018.03.083},
   DOI={10.1016/j.physletb.2018.03.083},
   journal={Physics Letters B},
   publisher={Elsevier BV},
   author={de Leeuw, Marius and Kristjansen, Charlotte and Linardopoulos, Georgios},
   year={2018},
   month=jun, pages={238–243} }

@article{de_Leeuw_2020,
   title={{Spin chain overlaps and the twisted Yangian}},
   volume={2020},
   ISSN={1029-8479},
   url={http://dx.doi.org/10.1007/JHEP01(2020)176},
   DOI={10.1007/jhep01(2020)176},
   number={1},
   journal={Journal of High Energy Physics},
   publisher={Springer Science and Business Media LLC},
   author={de Leeuw, Marius and Gombor, Tamás and Kristjansen, Charlotte and Linardopoulos, Georgios and Pozsgay, Balázs},
   year={2020},
   month=jan }

@misc{deleeuw2017introductionintegrabilityonepointfunctions,
      title={{Introduction to Integrability and One-point Functions in $\mathcal{N}=4$ SYM and its Defect Cousin}}, 
      author={M. de Leeuw and A. C. Ipsen and C. Kristjansen and M. Wilhelm},
      year={2017},
      eprint={1708.02525},
      archivePrefix={arXiv},
      primaryClass={hep-th},
      url={https://arxiv.org/abs/1708.02525}, 
}

@article{Pozsgay_2014,
   title={{Overlaps between eigenstates of the XXZ spin-1/2 chain and a class of simple product states}},
   volume={2014},
   ISSN={1742-5468},
   url={http://dx.doi.org/10.1088/1742-5468/2014/06/P06011},
   DOI={10.1088/1742-5468/2014/06/p06011},
   number={6},
   journal={Journal of Statistical Mechanics: Theory and Experiment},
   publisher={IOP Publishing},
   author={Pozsgay, Balázs},
   year={2014},
   month=jun, pages={P06011} }

@article{Brockmann_2014,
   title={{A Gaudin-like determinant for overlaps of Néel and XXZ Bethe states}},
   volume={47},
   ISSN={1751-8121},
   url={http://dx.doi.org/10.1088/1751-8113/47/14/145003},
   DOI={10.1088/1751-8113/47/14/145003},
   number={14},
   journal={Journal of Physics A: Mathematical and Theoretical},
   publisher={IOP Publishing},
   author={Brockmann, M and De Nardis, J and Wouters, B and Caux, J-S},
   year={2014},
   month=mar, pages={145003} }

@article{Jiang_2020,
   title={{Structure constants in $ \mathcal{N} $ = 4 SYM at finite coupling as worldsheet g-function}},
   volume={2020},
   ISSN={1029-8479},
   url={http://dx.doi.org/10.1007/JHEP07(2020)037},
   DOI={10.1007/jhep07(2020)037},
   number={7},
   journal={Journal of High Energy Physics},
   publisher={Springer Science and Business Media LLC},
   author={Jiang, Yunfeng and Komatsu, Shota and Vescovi, Edoardo},
   year={2020},
   month=jul }

@article{Jiang_2019,
   title={{Exact Three-Point Functions of Determinant Operators in Planar $ \mathcal{N} $ = 4 Supersymmetric Yang-Mills Theory}},
   volume={123},
   ISSN={1079-7114},
   url={http://dx.doi.org/10.1103/PhysRevLett.123.191601},
   DOI={10.1103/physrevlett.123.191601},
   number={19},
   journal={Physical Review Letters},
   publisher={American Physical Society (APS)},
   author={Jiang, Yunfeng and Komatsu, Shota and Vescovi, Edoardo},
   year={2019},
   month=nov }

@article{Minahan_2008,
   title={{The Bethe ansatz for superconformal Chern-Simons}},
   volume={2008},
   ISSN={1029-8479},
   url={http://dx.doi.org/10.1088/1126-6708/2008/09/040},
   DOI={10.1088/1126-6708/2008/09/040},
   number={09},
   journal={Journal of High Energy Physics},
   publisher={Springer Science and Business Media LLC},
   author={Minahan, J.A and Zarembo, K},
   year={2008},
   month=sep, pages={040–040} }

@article{Bak_2008,
   title={{Integrable spin chain in superconformal Chern-Simons theory}},
   volume={2008},
   ISSN={1029-8479},
   url={http://dx.doi.org/10.1088/1126-6708/2008/10/053},
   DOI={10.1088/1126-6708/2008/10/053},
   number={10},
   journal={Journal of High Energy Physics},
   publisher={Springer Science and Business Media LLC},
   author={Bak, Dongsu and Rey, Soo-Jong},
   year={2008},
   month=oct, pages={053–053} }

@article{Yang_2022,
   title={{Three-point functions in ABJM and Bethe Ansatz}},
   volume={2022},
   ISSN={1029-8479},
   url={http://dx.doi.org/10.1007/JHEP01(2022)002},
   DOI={10.1007/jhep01(2022)002},
   number={1},
   journal={Journal of High Energy Physics},
   publisher={Springer Science and Business Media LLC},
   author={Yang, Peihe and Jiang, Yunfeng and Komatsu, Shota and Wu, Jun-Bao},
   year={2022},
   month=jan }

@article{yang2024integrableboundarystatesmaximal,
title = {{Integrable boundary states from maximal giant gravitons in ABJM theory}},
author = {Peihe Yang},
journal = {Physics Letters B},
volume = {846},
pages = {138194},
year = {2023},
issn = {0370-2693},
doi = {https://doi.org/10.1016/j.physletb.2023.138194},
url = {https://www.sciencedirect.com/science/article/pii/S0370269323005282},
}

@article{Kristjansen_2022,
   title={{Integrable domain walls in ABJM theory}},
   volume={2022},
   ISSN={1029-8479},
   url={http://dx.doi.org/10.1007/JHEP02(2022)070},
   DOI={10.1007/jhep02(2022)070},
   number={2},
   journal={Journal of High Energy Physics},
   publisher={Springer Science and Business Media LLC},
   author={Kristjansen, Charlotte and Vu, Dinh-Long and Zarembo, Konstantin},
   year={2022},
   month=feb }

@article{Jiang_2020_overlap,
   title={{On exact overlaps in integrable spin chains}},
   volume={2020},
   ISSN={1029-8479},
   url={http://dx.doi.org/10.1007/JHEP06(2020)022},
   DOI={10.1007/jhep06(2020)022},
   number={6},
   journal={Journal of High Energy Physics},
   publisher={Springer Science and Business Media LLC},
   author={Jiang, Yunfeng and Pozsgay, Balázs},
   year={2020},
   month=jun }

@misc{jiang2025rationalqsystemsintegrablespin,
      title={{Rational $Q$-systems for integrable spin chains without $U(1)$ symmetry}}, 
      author={Yunfeng Jiang and Yi-Chao Liu and Yuan Miao and Zi-Xi Tan},
      year={2025},
      eprint={2512.01551},
      archivePrefix={arXiv},
      primaryClass={hep-th},
      url={https://arxiv.org/abs/2512.01551}, 
}

@article{Caetano_2022,
   title={{Crosscap States in Integrable Field Theories and Spin Chains}},
   volume={187},
   ISSN={1572-9613},
   url={http://dx.doi.org/10.1007/s10955-022-02914-6},
   DOI={10.1007/s10955-022-02914-6},
   number={3},
   journal={Journal of Statistical Physics},
   publisher={Springer Science and Business Media LLC},
   author={Caetano, João and Komatsu, Shota},
   year={2022},
   month=apr }

@article{He_2023,
   title={{Integrable crosscap states: from spin chains to 1D Bose gas}},
   volume={2023},
   ISSN={1029-8479},
   url={http://dx.doi.org/10.1007/JHEP08(2023)079},
   DOI={10.1007/jhep08(2023)079},
   number={8},
   journal={Journal of High Energy Physics},
   publisher={Springer Science and Business Media LLC},
   author={He, Miao and Jiang, Yunfeng},
   year={2023},
   month=aug }

@article{TakFad79,
  author       = {L. A. Takhtadzhyan and L. D. Faddeev},
  title        = {{The quantum method of the inverse problem and the Heisenberg {XYZ} model}},
  journal      = {Russian Mathematical Surveys},
  year         = {1979},
  volume       = {34},
  number       = {5},
  pages        = {11--68},
  doi          = {10.1070/RM1979v034n05ABEH003909},
  url          = {http://mi.mathnet.ru/eng/rm4115},
  note         = {MR562799}
}

@article{Slavnov_2020,
   title={{Scalar products of Bethe vectors in the 8-vertex model}},
   volume={2020},
   ISSN={1029-8479},
   url={http://dx.doi.org/10.1007/JHEP06(2020)123},
   DOI={10.1007/jhep06(2020)123},
   number={6},
   journal={Journal of High Energy Physics},
   publisher={Springer Science and Business Media LLC},
   author={Slavnov, N. and Zabrodin, A. and Zotov, A.},
   year={2020},
   month=jun }

@article{PhysRevLett.26.832,
  title = {{Eight-Vertex Model in Lattice Statistics}},
  author = {Baxter, R. J.},
  journal = {Phys. Rev. Lett.},
  volume = {26},
  issue = {14},
  pages = {832--833},
  numpages = {0},
  year = {1971},
  month = {Apr},
  publisher = {American Physical Society},
  doi = {10.1103/PhysRevLett.26.832},
  url = {https://link.aps.org/doi/10.1103/PhysRevLett.26.832}
}

@article{Gombor_2021,
   title = {{On factorized overlaps: Algebraic Bethe Ansatz, twists, and separation of variables}},
   volume={967},
   ISSN={0550-3213},
   url={http://dx.doi.org/10.1016/j.nuclphysb.2021.115390},
   DOI={10.1016/j.nuclphysb.2021.115390},
   journal={Nuclear Physics B},
   publisher={Elsevier BV},
   author={Gombor, Tamás and Pozsgay, Balázs},
   year={2021},
   month=jun, pages={115390} }

@article{gombor2025exactoverlapsintegrablematrix,
  author    = {Gombor, T. and Kristjansen, C. and Moustakis, V. and others},
  title     = {{On exact overlaps of integrable matrix product states: inhomogeneities, twists and dressing formulas}},
  journal   = {Journal of High Energy Physics},
  volume    = {2025},
  number    = {2},
  pages     = {100},
  year      = {2025},
  doi       = {10.1007/JHEP02(2025)100},
  url       = {https://doi.org/10.1007/JHEP02(2025)100}
}

@article{Vega_1993,
   title={{Boundary K-matrices for the six vertex and the n(2n-1)An-1vertex models}},
   volume={26},
   ISSN={1361-6447},
   url={http://dx.doi.org/10.1088/0305-4470/26/12/007},
   DOI={10.1088/0305-4470/26/12/007},
   number={12},
   journal={Journal of Physics A: Mathematical and General},
   publisher={IOP Publishing},
   author={Vega, H J de and Ruiz, A Gonzalez},
   year={1993},
   month=jun, pages={L519–L524} }

@article{Zhang_2015,
doi = {10.1088/1742-5468/2015/05/P05014},
url = {https://doi.org/10.1088/1742-5468/2015/05/P05014},
year = {2015},
month = {may},
publisher = {IOP Publishing and SISSA},
volume = {2015},
number = {5},
pages = {P05014},
author = {Zhang, Xin and Li, Yuan-Yuan and Cao, Junpeng and Yang, Wen-Li and Shi, Kangjie and Wang, Yupeng},
title = {Retrieve the Bethe states of quantum integrable models solved via the off-diagonal Bethe Ansatz},
journal = {Journal of Statistical Mechanics: Theory and Experiment}
}

@article{BAXTER1973_1,
title = {{Eight-vertex model in lattice statistics and one-dimensional anisotropic heisenberg chain. I. Some fundamental eigenvectors}},
journal = {Annals of Physics},
volume = {76},
number = {1},
pages = {1-24},
year = {1973},
issn = {0003-4916},
doi = {https://doi.org/10.1016/0003-4916(73)90439-9},
url = {https://www.sciencedirect.com/science/article/pii/0003491673904399},
author = {Rodney Baxter}
}

@article{KITANINE1999,
title = {{Form factors of the XXZ Heisenberg spin-$\frac12$ finite chain}},
author = {N. Kitanine and J.M. Maillet and V. Terras},
journal = {Nucl. Phys. B},
volume = {554},
number = {3},
pages = {647-678},
year = {1999},
issn = {0550-3213},
doi = {https://doi.org/10.1016/S0550-3213(99)00295-3},
url ={https://www.sciencedirect.com/science/article/pii/S0550321399002953}
}

@article{Jepsen:2021,
    author = "Jepsen, Paul Niklas and Lee, Yoo Kyung {\textquoteleft}Eunice{\textquoteright} and Lin, Hanzhen and Dimitrova, Ivana and Margalit, Yair and Ho, Wen Wei and Ketterle, Wolfgang",
    title = "{Long-lived phantom helix states in Heisenberg quantum magnets}",
    eprint = "2110.12043",
    archivePrefix = "arXiv",
    primaryClass = "cond-mat.quant-gas",
    doi = "10.1038/s41567-022-01651-7",
    journal = "Nature Phys.",
    volume = "18",
    number = "8",
    pages = "899--904",
    year = "2022"
}

@article{Jepsen:2021_2,
    author = "Jepsen, Paul Niklas and Ho, Wen Wei and Amato-Grill, Jesse and Dimitrova, Ivana and Demler, Eugene and Ketterle, Wolfgang",
    title = "{Transverse spin dynamics in the anisotropic Heisenberg model realized with ultracold atoms}",
    eprint = "2103.07866",
    archivePrefix = "arXiv",
    primaryClass = "cond-mat.quant-gas",
    doi = "10.1103/PhysRevX.11.041054",
    journal = "Phys. Rev. X",
    volume = "11",
    pages = "041054",
    year = "2021"
}

@article{sklyanin1989,
  title={Separation of variables in the Gaudin model},
  author={Sklyanin, EK},
  journal={J. Sov. Math.},
  volume={47},
  number={2},
  pages={2473--2488},
  year={1989},
  publisher={Springer}
}

@book{takahashi1999,
  title={Thermodynamics of one-dimensional solvable models},
  author={Takahashi, Minoru},
  year={1999},
  publisher={Cambridge university press Cambridge}
}

@misc{Yuan_2024,
      title={Towers of Quantum Many-body Scars from Integrable Boundary States}, 
      author={Kazuyuki Sanada and Yuan Miao and Hosho Katsura},
      year={2024},
      eprint={2411.01270},
      archivePrefix={arXiv},
      primaryClass={cond-mat.stat-mech},
      url={https://arxiv.org/abs/2411.01270}, 
}

@article{Nepomechie_2024,
   title={Spin‐s Dicke States and Their Preparation},
   volume={7},
   ISSN={2511-9044},
   url={http://dx.doi.org/10.1002/qute.202400057},
   DOI={10.1002/qute.202400057},
   number={12},
   journal={Advanced Quantum Technologies},
   publisher={Wiley},
   author={Nepomechie, Rafael I. and Ravanini, Francesco and Raveh, David},
   year={2024},
   month=sep }

@article{Pozsgay_2018,
   title={Overlaps with arbitrary two-site states in the XXZ spin chain},
   volume={2018},
   ISSN={1742-5468},
   url={http://dx.doi.org/10.1088/1742-5468/aabbe1},
   DOI={10.1088/1742-5468/aabbe1},
   number={5},
   journal={Journal of Statistical Mechanics: Theory and Experiment},
   publisher={IOP Publishing},
   author={Pozsgay, B},
   year={2018},
   month=may, pages={053103} }

@article{BOYU_1993,
title = {{General solution of the reflection equations for the eight-vertex model}},
journal = {Physics Letters A},
volume = {183},
number = {2},
pages = {169-174},
year = {1993},
issn = {0375-9601},
doi = {https://doi.org/10.1016/0375-9601(93)91165-2},
url = {https://www.sciencedirect.com/science/article/pii/0375960193911652},
author = { {Bo-Yu Hou} and Rui-Hong Yue}
}

@article{Inami_1994,
   title={{Integrable XYZ spin chain with boundaries}},
   volume={27},
   ISSN={1361-6447},
   url={http://dx.doi.org/10.1088/0305-4470/27/24/002},
   DOI={10.1088/0305-4470/27/24/002},
   number={24},
   journal={Journal of Physics A: Mathematical and General},
   publisher={IOP Publishing},
   author={Inami, T and Konno, H},
   year={1994},
   month=dec, pages={L913–L918} }

@article{Guan1996,
Author = {Xi-Wen Guan, Dian-Min Tong,Huan-Qiang Zhou},
Title = {Lax pair formulation for the one-dimensional Heisenberg XYZ open chain},
Journal = {J. Phys. Soc. Jpn.},
Year = {1996},
Volume = {65},
Number = {9},
Pages = {2807-2810},
DOI = {10.1143/JPSJ.65.2807}
}

@article{He_2025,
   title={Spin-$s$ $Q$-systems: Twist and open boundaries},
   volume={8},
   ISSN={2666-9366},
   url={http://dx.doi.org/10.21468/SciPostPhysCore.8.3.057},
   DOI={10.21468/scipostphyscore.8.3.057},
   number={3},
   journal={SciPost Physics Core},
   publisher={Stichting SciPost},
   author={He, Yi-Jun and Hou, Jue and Liu, Yi-Chao and Tan, Zi-Xi},
   year={2025},
   month=aug }


\end{document}